\documentclass[journal]{IEEEtran}
\usepackage{blindtext}
\usepackage{balance}		
\usepackage{cite}
\usepackage[pdftex]{graphicx}
\usepackage[cmex10]{amsmath}
\usepackage{amssymb}
\usepackage{siunitx}
\usepackage[tight,footnotesize]{subfigure}
\usepackage{url}
\usepackage{mathrsfs}
\usepackage{amsthm}

\usepackage{multirow}
\usepackage{amsmath}
\usepackage{cleveref}
\DeclareMathOperator*{\argmin}{arg\,min}
\usepackage{amsmath,amssymb,amsfonts}
\usepackage{amsthm}
\usepackage{siunitx}
\usepackage{algorithmic}
\usepackage{graphicx}
\usepackage{textcomp}
\usepackage{mathtools}
\usepackage{multirow}
\usepackage{mathrsfs}
\usepackage{cleveref}
\usepackage[tight,footnotesize]{subfigure}
\newtheorem{theorem}{Theorem}
\newtheorem{lem}{Lemma}

\newtheorem{corollary}{Corollary}

\allowdisplaybreaks

\usepackage{xcolor}
\definecolor{morange}{rgb}{0.8,0.2,0}
\definecolor{mblue}{rgb}{0,0.3,1.0}
\definecolor{mgreen}{rgb}{0.2,0.4,0}

\begin{document}
\bstctlcite{IEEEexample:BSTcontrol}

\title{Analog Self-Interference Cancellation with Practical RF Components for Full-Duplex Radios}

\author{Jong Woo Kwak,~\IEEEmembership{Student Member,~IEEE,}
        Min Soo Sim, In-Woong Kang, \\Jaedon Park, Kai-Kit Wong,~\IEEEmembership{Fellow,~IEEE}, and~Chan-Byoung Chae,~\IEEEmembership{Fellow,~IEEE}%
\thanks{J. W. Kwak and C.-B. Chae are with School of Integrated Technology, Yonsei University, Seoul, Korea. M. S. Sim was with School of Integrated Technology, Yonsei University, Seoul, Korea. He is now with Qualcomm Technologies, Inc., San Diego, CA, USA (e-mail: kjw8216@yonsei.ac.kr; cbchae@yonsei.ac.kr; msim@qti.qualcomm.com).}
\thanks{I.-W. Kang and J. Park are with Agency for Defense Development, Daejeon, Korea (e-mail: iwkang@add.re.kr; jaedon2@add.re.kr).}
\thanks{K. Wong is with University College London, U.K (e-mail: kai-kit.wong@ucl.ac.uk).}

}

\markboth{}%
{}


\maketitle

\begin{abstract}
One of the main obstacles in full-duplex radios is analog-to-digital converter (ADC) saturation on a receiver due to the strong self-interference (SI). 
To solve this issue, researchers have proposed two different types of analog self-interference cancellation (SIC) methods---i) passive suppression and ii) regeneration-and-subtraction of SI. 
For the latter case, the tunable RF component, such as a multi-tap circuit, reproduces and subtracts the SI. The resolutions of such RF components constitute the key factor of the analog SIC. Indeed, they are directly related to how well the SI is imitated.
Another major issue in analog SIC is the inaccurate estimation of the SI channel due to the nonlinear distortions, which mainly come from the power amplifier (PA). 
In this paper, we derive a closed-form expression for the SIC performance of the multi-tap circuit; we consider how the RF components must overcome such practical impairments as digitally-controlled attenuators, phase shifters, and PA. 
For a realistic performance analysis, we exploit the measured PA characteristics and carry out a 3D ray-tracing-based, system-level throughput analysis. 
Our results confirm that the non-idealities of the RF components significantly affect the analog SIC performance.
We believe our study provides insight into the design of the practical full-duplex system.
\end{abstract}


%
\IEEEpeerreviewmaketitle

\section{Introduction}

\IEEEPARstart{T}{o} keep up with the ever-increasing demand of higher data rates and spectral efficiency, for 5G wireless communications, full-duplex communications has emerged as a highly promising technique~\cite{dgsurvey, ibsurvey,hongmag,kwak_mole,Opp_Kim,fairmac,IAB_Suk,guest}.  
The most significant hurdle in full-duplex radio is the fact that an analog-to-digital converter (ADC) has to convert, simultaneously, a signal-of-interest (SoI) and the self-interference (SI). Because the power of SI typically exceeds that of SoI by 100 dB~\cite{passive}, successful conversion is possible only if the SI is suppressed within the dynamic range of the ADC by an analog SI cancellation (SIC).

For sub-6 GHz full-duplex radios, the passive SI suppressors (e.g., circulator, polarized antenna, etc.) typically do not provide a sufficient SIC capability to prevent the saturation of the receiver ADC.
Therefore, in full-duplex systems engineers widely adopt additional analog SIC methods~\cite{allanalog, passive,practical_realtime,kwakasilomar,Alms1,Alms2}. 
Such methods utilize the additional RF components (e.g., a multi-tap circuit or auxiliary transmit chain) to regenerate the canceling signal using the known transmitted signal and the SI channel knowledge~\cite{theoretical,springer,Sachin2013Full,lincoln}.
The multi-tap circuit consists of tunable RF components such as time delays, attenuators, and phase shifters. The value of each component is adjusted to match the circuit's output to the negative of the SI. Because the circuit takes a power amplifier (PA) output as an input, the analog SIC via multi-tap circuit is capable of eliminating the transmitter noise. 

The SIC capability of the multi-tap circuit depends on the 1) accurate estimation and 2) replication of the received SI. 
In~\cite{lincoln}, the authors derived an optimal attenuators/phase shifters values that minimized the residual SI with the fixed time-delay values. 
Note that the attenuators and phase shifters are assumed to have infinite resolutions. The resolutions of digitally-controlled attenuators and phase shifters are often neglected in the tuning algorithm, as it makes the problem NP-hard. The authors in~\cite{Sachin2013Full,MIMO,DLS} also pointed this out, and eased the problem through linear relaxation and the gradient-descent-based search.  
In~\cite{theoretical}, the authors theoretically analyzed the power of the residual SI with the proposed tuning algorithm in~\cite{lincoln}, considering the imperfect SI channel estimation and the imperfect time delay alignment between the SI channel and the circuit. 
However, this work also assumes the ideal attenuators and phase shifters. 
Through simulations, researchers in~\cite{postech,sparse,atten_quantization} investigated the impact of attenuator resolution on SIC performance. Their results showed that an attenuator's low resolution severely degraded the SIC performance. 
In~\cite{postech}, the authors observed that the SIC performance is rapidly improved as they increased the resolutions of phase shifters and attenuators, however, it is then saturated when the resolutions reached to the certain level (i.e., 0.05 dB stepsize for the attenuator and 10-bit phase shifter). According to the simulation results, the authors in~\cite{postech} set the resolutions of phase shifters and attenuators that provide sufficient analog SIC to prevent the ADC saturation.  
The effect of attenuator bias, response time and the phase noise introduced by attenuator is investigated in~\cite{atten_quantization,phasenoise,atten_bias_response}.

The other crucial issue affecting the analog SIC via multi-tap circuit performance is the transmitter nonlinearity~\cite{AUX_nonlin,nonlinDSIC, Alms1, Alms2}.
The authors in~\cite{Alms0,Alms1} proposed the purely analog hardware that cancels the SI and analyzed the impact of I/Q imbalance. In this paper, we focus on the nonlinearity caused by a PA. 
Although the SI channel estimation can be done offline at the initial, it has to be updated according to how the SI channel changes. In~\cite{MIMO}, the authors implemented WiFi-based MIMO full-duplex radios and observed that the incorrect SI channel estimation due to the PA nonlinearity limits SIC performance to 30 dB. They proposed an iterative tuning algorithm that resolves this problem and achieve 60 dB cancellation. 

Taking account into the practical issues mentioned above, this paper provides a realistic performance analysis of the multi-tap circuit.
The main contributions of this paper are as follows: 
\begin{itemize}
	\item A closed-form expression of the multi-tap circuit's SIC performance is derived while considering the resolution of the digital attenuators, phase shifters, and the practical PA characteristics. 
	To the best of our knowledge, this is the first work that theoretically investigates the analog SIC performance with such RF non-idealities. The derived formula is verified through link-level SIC simulations. Our results show that adaptive processing for the tuning of the multi-tap circuit could be necessary when the resolutions of the phase shifters and attenuators are low. Without extensive numerical evaluations, we can investigate the performance of the multi-tap circuit with various parameters by utilizing the derived expression. Future work would analyze the trade-off between the level of cancellation and the complexity cost of the multi-tap circuit.
	\item  
	We carry system-level throughput analysis considering UE-to-UE interference. We model the 3D building structure and adapt it to the ray-tracing tools with the measured radiation pattern of the dual-polarized antenna~\cite{prototyping,dualpole}. We compare the system-level throughputs in the various SIC scenarios. 
\end{itemize}
 
The rest of this paper is organized as follows. In Section~2, we introduce our system model and the preliminaries on the analog SIC via multi-tap circuit. In Section~3 and Section~4, we derive the closed-form expressions for the residual SI considering the limited resolution of the attenuators and the phase shifters and the nonlinearity of PA. In Section~5, we present the link-level SIC simulation results and the system-level throughput analysis. Finally, in Section~6, we present our conclusions.

\section{{\fontsize{11}{14}\selectfont System Model and\\  Multi-Tap Circuit-Based SIC}}
\label{Sec.model}

\begin{figure}[t]
	\begin{center}
		\resizebox{2.8in}{!}{\includegraphics{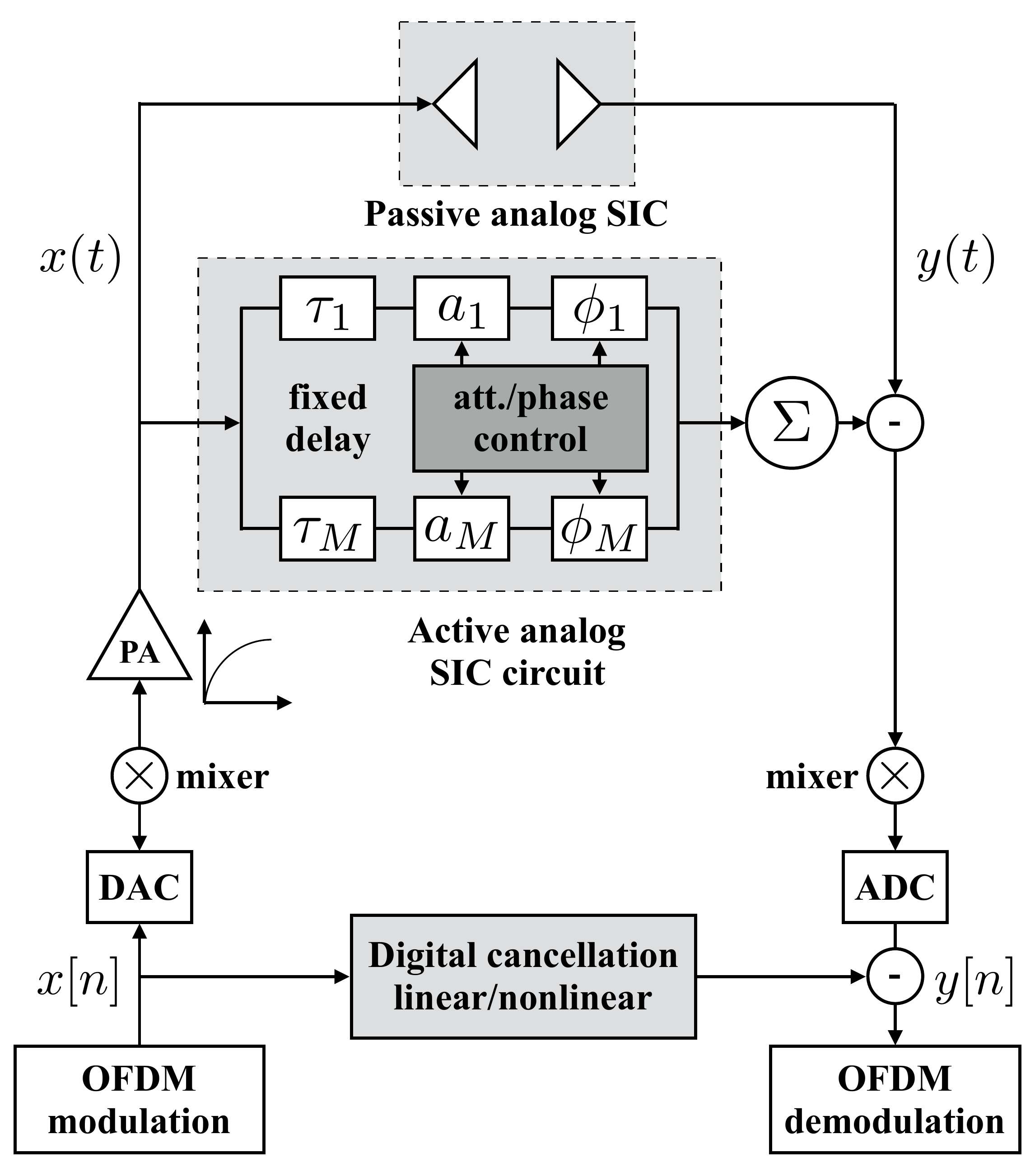}}
		\caption{A block diagram of the self-interference cancellation simulator.}	
		\label{block_diagram}
	\end{center}
\end{figure}

Fig.~\ref{block_diagram} depicts a full-duplex system equipped with a analog SIC via multi-tap circuit. 
We consider the following three-step SIC scenario: 1) the passive suppression at the propagation stage 2) the analog SIC via multi-tap circuit, and 3) the digital cancellation of the residual SI. The SI channel herein refers to the response of leakage from a transmitter to the receiver (i.e., line-of-sight component) and the reflected signals (non-line-of-sight components). 
Let $h_\text{SI}(t)$ be a baseband equivalent impulse response of the SI channel, then,
\begin{equation}
\label{eq.h_SI}
h_\text{SI}(t) = \sum_{i=0}^{L-1}c_{i}\delta(t-iT),
\end{equation} 
where $c_i$ is the gain of the $i$-th tap, $L$ is the number of taps, and $T$ is the baseband sampling period. For the first tap coefficient, we adopt a Rician channel model~\cite{rician} and a Raleigh channel model for the other taps. The Rician factor is 20dB. We consider an OFDM system with $K$ subcarriers, where the frequency-domain channel gain for $k$-th subcarrier is represented as
\begin{equation} 
\label{h_si_freq}
{{H_{\text{SI}}}}[k]=\sum_{\ell=0}^{L-1}c_ie^{-j2\pi k\ell /K}.
\end{equation}
Let $a_i$, $\phi_i$, and $\tau_i$ be the attenuator, phase shifter, and time delay value of the $i$-th tap, respectively, and $\bold{{H}^{i}}$ is the corresponding frequency response of the $i$-th tap. The time delays of the multi-tap circuit herein are assumed to be pre-determined\cite{lincoln,Sachin2013Full}. 
Note that still under investigation is the optimal setting of the circuit's time-delay configuration~\cite{numoftap}. For the derivation, we consider a general time-delay setting.
We assume the delay line will not interfere with other delay lines.
The frequency response of the multi-tap circuit is modeled as a summation of each delay line's frequency response,
\begin{equation}
{H_\text{cir}}[k]=\sum_{i=1}^{M}{H^{i}}[k],
\end{equation}
where
\begin{equation}
\label{eq_circuit}
{H^{i}}[k]= a_ie^{-j\phi_i}e^{-jk\Delta_w\tau_i}. 
\end{equation}
Our goal is to match the frequency response of the multi-tap circuit to the negative of the estimated SI channel, $\bold{\bold{\hat{H}_{\text{SI}}}}$, 
\begin{equation}
\label{h_noise}
{{\hat{H}_{\text{SI}}}}[k] ={{H_{\text{SI}}}}[k]+{N}[k],
\end{equation}
where ${N}[k]$ is the circular symmetric complex Gaussian (CSCG) noise with zero mean and variance $\sigma^2$.
To this end, we adjust the values of variable attenuators and phase shifters by solving the following optimization problem.
\begin{align}
\label{eq.h_circuit}
&\Big\{{\tilde{a},\tilde{\phi}}\Big\}=\underset{\{a_i, \phi_i\} }{\text{argmin}}   \sum_{k=0}^{K-1}\left({{\hat{H}_{\text{SI}}}}[k]-{H_{\text{cir}}}[k])\right)^2. 
\end{align}
Let $\bold{W}$ denotes the coefficients vector, which contains the attenuator and phase shifter values.
\begin{align}
\label{eq.wo}
&\bold{W} = \Big[{a}_1e^{-j{\phi}_1},{a}_2e^{-j{\phi}_2},
\cdots, {a}_Me^{-j{\phi}_M} \Big]^T, 
\end{align} 
where $(\cdot)^T$ is the matrix transposition operation.
With the coefficients $\bold{W}$, the corresponding frequency response of the multi-tap circuit can be represented as
\begin{align}
\label{eq.h_circuit2}
\bold{H_\text{cir}}&=\underbrace{\begin{bmatrix}
	e^{-j\Delta_w\tau_1} & e^{-j\Delta_w\tau_2} & \cdots & e^{-j\Delta_w\tau_M}\\
	e^{-j2\Delta_w\tau_1} & e^{-j2\Delta_w\tau_2} & \cdots & e^{-j2\Delta_w\tau_M}\\
	\cdots & \cdots & \ddots& \cdots\\
	e^{-jK\Delta_w\tau_1}& e^{-jK\Delta_w\tau_2} & \cdots & e^{-jK\Delta_w\tau_M}
	\end{bmatrix}}_{\bold{\bold{\Omega}}} \bold{W}, \nonumber\\
\end{align}
where $\Delta_w$ is the sampling interval in the frequency domain.
Note that the $i$-th column of $\bold{\bold{\Omega}}$ indicates the frequency response of the time delay $\tau_i$.
The Wiener solution $\bold{\bold{W_o}}$ of~\eqref{eq.h_circuit} is derived in~\cite{lincoln} as
\begin{align}
\label{eq.wiener}
\bold{\bold{W_o}} &= \left(\bold{\bold{\Omega}}^*\bold{\bold{\Omega}}\right)^{-1}\bold{\bold{\Omega}}^*\bold{\bold{\hat{H}_{\text{SI}}}}, \\\nonumber
&=\bold{R}^{-1}\bold{\bold{\Omega}}^*\bold{\bold{\hat{H}_{\text{SI}}}},
\end{align} 
where $(\cdot)^*$ is the matrix Hermitian operation and $\bold{R}=\bold{\bold{\Omega}}^*\bold{\bold{\Omega}}$.
By substituting~\eqref{eq.wiener} into (8), we get the frequency response of the optimized circuit $\bold{H_{\text{cir}}^o}$, 
\begin{align}
\label{hciro}
\bold{H_{\text{cir}}^o}=\bold{\Omega} \bold{W_o} = \bold{\Omega} \bold{R}^{-1}\bold{\Omega}^*\bold{\hat{H}_{\text{SI}}}.
\end{align}
The effective SI channel including the multi-tap circuit can be expressed as $\bold{H_\text{SI}}-\bold{H_\text{cir}^o}$.
The average power of the effective SI channel is then represented as follows:
\begin{align}
\label{eq.residual}
P_{H_{\text{eff}}}&=\frac{1}{K}\mathbb{E}\left[\text{tr}\left[(\bold{H_\text{SI}}-\bold{H_\text{cir}^o})(\bold{H_\text{SI}}-\bold{H_\text{cir}^o})^*\right]\right], 
\end{align} 
where $\text{tr}(\cdot)$ denotes the trace of a matrix.  
The authors in~\cite{theoretical} derived $P_{H_{\text{eff}}}$ as follows:
\begin{align}
\label{eq.residual2}
P_{H_{\text{eff}}} &= \frac{M}{K}\sigma^2 + \frac{1}{K}\text{tr}(\bold{H_{\text{SI}}}\bold{H_{\text{SI}}}^*-\bold{\Omega} \bold{R}^{-1}\bold{\Omega}^*\bold{H_{\text{SI}}}\bold{H_{\text{SI}}}^*).
\end{align} 
Note the authors in ~\cite{theoretical} assumed a deterministic SI channel. 
The analog SIC can be represented as $C=-10\log_{10}(P_{H_{\text{eff}}})$. In this paper, we denote the estimation error of the SI channel induced by noise, $r$, as follows:
\begin{align}
r=10\log_{10}\left(\frac{\sigma^2}{\text{Power}(\bold{H_{\text{SI}}})}\right).
\end{align}
Equation~\eqref{eq.residual2} shows that the multi-tap circuit's SIC performance heavily depends on the time delays of the circuit and the estimation error of the SI channel. 
The results in~\cite{theoretical}, however, neglected the quantization errors induced by the phase shifters and attenuators.  
Moreover, the estimation errors of the SI channel caused by the nonlinear distortions were simply treated as Gaussian noise, where the appropriate values of the mean and variance of the noise was not introduced. The authors in~\cite{MIMO} observed that the nonlinear distortions significantly reduce the multi-tap circuit's SIC performance. 

In the rest of this paper, we show that the resolutions of the phase shifters and attenuators are crucial for the analog SIC via multi-tap circuit. We also analyze the impact of the nonlinear distortions on the multi-tap circuit's SIC performance more precisely. 



\section{Impact of Non-Ideal RF Components}
\subsection{Non-Ideal Attenuators and Phase Shifters}
\label{Sec.evaluation}

The attenuators and phase shifters are extensively adopted in the analog SIC, where the additional RF components regenerate the negative of SI. The authors in~\cite{atten_quantization,atten_quantization2} the impact of practical impairments of the attenuators and phase shifters such as phase-shift introduced by an attenuator are investigated. The quantization errors induced at the phase shifters and attenuators are considered in~\cite{atten_quantization, ian_codebook}. In~\cite{ian_codebook}, the authors proposed a mmWave beamforming codebook that minimizes the SI while achieving high beamforming gain~\cite{mmwavelens} over the desired coverage regions. The authors in~\cite{atten_quantization} numerically analyzed the multi-tap circuit's SIC performance with a finite attenuator stepsize.  

In this section, we derive the multi-tap circuit's SIC performance with the $B$-bit phase shifters and the attenuator with a stepsize of $\delta$ in dB scale. The values of attenuation $a_i$ and phase shift $\phi_i$ are then quantized as
\begin{align}
-&20\log_{10} a_i \in \left\{0, \delta, 2\delta, \cdots \right\}, \nonumber\\
-\phi &\in \left\{0, \frac{2\pi}{2^B}, \frac{4\pi}{2^B}, \cdots, \frac{(2^B-1)2\pi}{2^B} \right\}. \nonumber
\end{align}
Let $n_{a,i}$ and $n_{q,i}$ denote the quantization errors induced by the $i$-th attenuator and phase shifter. The quantization errors $n_{a,i}$ and $n_{q,i}$ are modeled as the uniform random variables ($\pmb{U}(\cdot)$),
\begin{align}
\label{eq.uniform}
&n_{a,i} \sim \pmb{U}\left(\left[-\frac{\delta}{2},\frac{\delta}{2}\right]\right), \\
&{n_{q,i}} \sim \pmb{U}\left(\left[-\frac{2\pi}{2^{B+1}},\frac{2\pi}{2^{B+1}}\right]\right).
\end{align}
The quantized attenuation and phase shifter $\tilde{a}^q_i$ and $\tilde{\phi}^q_i$ are then represented as 
\begin{align}
\label{quantization}
\tilde{a}^q_i &= \tilde{a}_i 10^{n_{a,i}/20}, \\
\tilde{\phi}^q_i &= \tilde{\phi}_i + n_{p,i}.
\end{align}

Now we define the quantization error matrix $\bold{Q}$, to represent the quantized tap coefficients $\bold{W_o}^q$ as follows: 
\begin{align}
\bold{W_o}^q &= \left[\tilde{a}_1^qe^{-j\tilde{\phi}_1^q},\tilde{a}_2^qe^{-j\tilde{\phi}_2^q},
\cdots, \tilde{a}_M^qe^{-j\tilde{\phi}_M^q}\right]^T \nonumber\\
&=\bold{W_o}\bold{Q}, \nonumber
\end{align}
where
\begin{align}
\label{eq.quan_coeff}
Q_{ij} &= \begin{cases}
10^{n_{a,i}/20}e^{-jn_{p,i}} &\text{for } i=j \\
0 &\text{for } i\neq j.
\end{cases} 
\end{align}
By replacing $\bold{W_o}$ as $\bold{W_o}^q$ in~\eqref{eq.h_circuit2} and~\eqref{eq.residual}, we obtain the average power of the effective SI channel with non-ideal attenuators and phase shifters ($P^q_{H_{\text{eff}}}$),
\begin{align}
\label{eq.residual_q}
P^q_{H_{\text{eff}}}&=\frac{1}{K}\mathbb{E}\left[\text{tr}\left[(\bold{H_\text{SI}}-\bold{H_\text{cir}^o}\bold{Q})(\bold{H_\text{SI}}-\bold{H_\text{cir}^o}\bold{Q})^*\right]\right]. 
\end{align} 
For convenience, we denote $\mathbb{E}\left[\bold{H_{{\mathrm{SI}}}H_{{\mathrm{SI}}}}^*\right]$ as $\bold{E_{H_{{\mathrm{SI}}}H_{{\mathrm{SI}}}^*}}$. 
\begin{figure*}[!t]
	\begin{center}	
		\subfigure[]	{\includegraphics[width=0.90\columnwidth,keepaspectratio]%
			{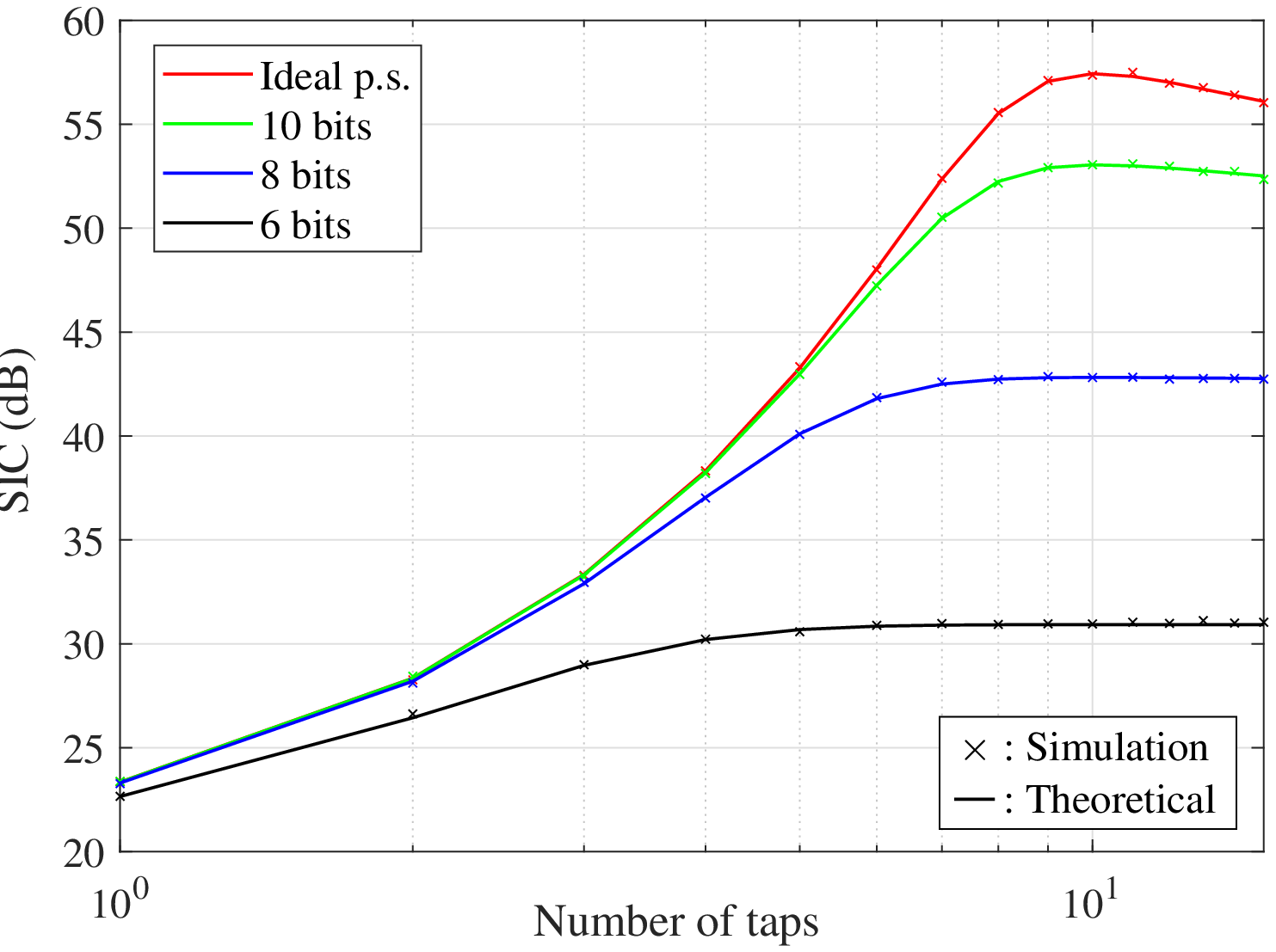}
			\label{fig_psbitcomp}}{\phantom{123}}
		\subfigure[]
		{\includegraphics[width=0.90\columnwidth,keepaspectratio]%
			{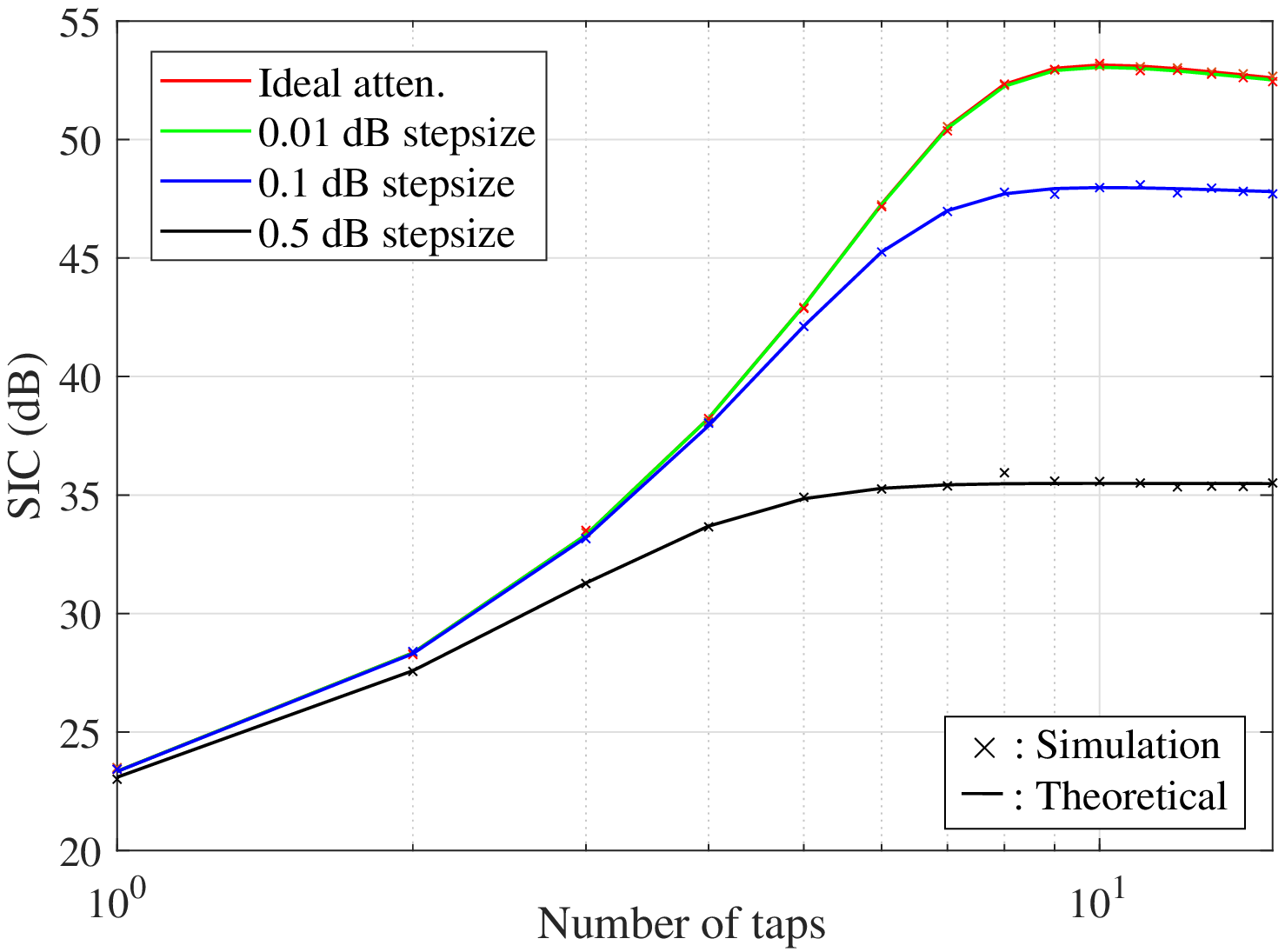}
			\label{fig_attendeltacomp}}
		\caption{SIC versus the number of taps  with different attenuator and phase shifter resolutions. The estimation error of SI channel induced by noise $r=-50~\text{dB}$. (a) The stepsize of the attenuator is fixed at 0.01 dB. (b) The phase shifter bits is fixed at 10.}
	\end{center}
\end{figure*}
\begin{theorem}
	The average power of the effective SI channel with the multi-tap circuit with $B$-bit phase shifters and attenuator stepsize $\delta$ (in dB scale), $P^q_{H_{\text{eff}}}$, can be rewritten as
	\begin{align}
	\label{eq.prop1}
	P^q_{H_{\text{eff}}} =& \frac{1}{K}{\mathrm{tr}}[\bold{E_{H_{{\mathrm{SI}}}H_{{\mathrm{SI}}}^*}}]+\frac{(PA_1)^2-2PA_1}{K}{\mathrm{tr}}[\bold{\Omega} \bold{R}^{-1}\bold{\Omega}^*\bold{E_{H_{{\mathrm{SI}}}H_{{\mathrm{SI}}}^*}}] \nonumber\\ 
	&+\frac{\sigma^2}{K} \left[(PA_1)^2M+(A_2-(PA_1)^2)K{\mathrm{tr}}(\bold{R}^{-1}) \right] \nonumber \\
	&+\left(A_2-(PA_1)^2\right){\mathrm{tr}}\left[\bold{\Omega}(\bold{R}^{-1})^2\bold{\Omega}^*\bold{E_{H_{{\mathrm{SI}}}H_{{\mathrm{SI}}}^*}} \right],
	\end{align}
	where
	$P=\frac{2^B}{\pi}\mathrm{sin} \left(\frac{\pi}{2^B} \right)$,
	$A_1=\frac{20}{\delta\mathrm{ln}(10)} \left(10^{\frac{\delta}{40}}-10^{\frac{-\delta}{40}} \right)$, and $A_2=\frac{10}{\delta\mathrm{ln}(10)} \left(10^{\frac{\delta}{20}}-10^{\frac{-\delta}{20}} \right)$.
\end{theorem}
\begin{proof}
	See Appendix-A.
\end{proof}	

\begin{corollary} 
	\label{corollary1}
	With ideal attenuators and phase shifters,~\eqref{eq.prop1} can be rewritten as~\eqref{eq.residual2}, i.e., 
	$\lim_{B\to\infty, \delta\to0}P^q_{H_{\text{eff}}} \!=\! P_{H_{\text{eff}}}.$
\end{corollary}
\begin{proof}
	By taking $B\to\infty$, and $\delta\to0$, we get the following:
	\begin{align}
	\label{lim_pa}
	\lim_{B\to\infty}P &=\lim_{B\to\infty}\frac{2^B}{\pi}\mathrm{sin}\left(\frac{\pi}{2^B}\right)=1, \ \  \nonumber \\
	\lim_{\delta \to 0}A_1 &=\lim_{\delta \to 0}\frac{20}{\delta\mathrm{ln}(10)}\left(10^{\frac{\delta}{40}}-10^{\frac{-\delta}{40}} \right)= 1,  \nonumber\\
	\lim_{\delta \to 0}A_2 &=\lim_{\delta \to 0}\frac{10}{\delta\mathrm{ln}(10)}\left(10^{\frac{\delta}{20}}-10^{\frac{-\delta}{20}} \right)= 1.
	\end{align} 
	We obtain Corollary~\ref{corollary1} by substituting $P\!=\!A_1\!=\!A_2\!=\!1$ in~\eqref{eq.prop1}.
\end{proof}

In~\eqref{eq.prop1}, $P$ and $(A_1,A_2)$ are associated with the resolutions of the phase shifters and attenuators, respectively.
Based on \eqref{eq.prop1}, we calculate the analog SIC via multi-tap circuit with non-ideal phase shifters and attenuators, $C^q=-10\log_{10}(P^q_{H_{\text{eff}}})$.
The SIC that provided by the multi-tap circuit with different phase shifter and attenuator resolutions are shown in Fig.~\ref{fig_psbitcomp} and Fig.~\ref{fig_attendeltacomp}. The simulation values are matched well with the theoretical values.  
The expected SI channel tap gains are set to $\{0~\text{dB}, -25~ \text{dB},-30~\text{dB}, \cdots, -75~\text{dB}\}$~\cite{theoretical} and the estimation error of the SI channel induced by noise is set to $-50$ dB. We transmit a 64QAM OFDM symbol with 64 subcarriers for the SI channel estimation.
For the SI channel estimation, all subcarriers are used.
To see the impact of the limited resolutions of phase shifters and attenuators, we set the delay of the $i$-th tap to that of the SI channel (i.e., $\tau_i=iT$).
The red line of Fig.~\ref{fig_psbitcomp} corresponds to the case of the ideal phase shifters. For the other three lines, we set the resolutions of the phase shifters to $\{10, 8, 6\}$ bits. The resolutions of the attenuators are fixed at 0.01 dB for Fig.~\ref{fig_psbitcomp}.
In Fig.~\ref{fig_attendeltacomp}, we compare the cases of the ideal attenuators and the attenuators with stepsize $\{0.01,0.1,0.5\}$ dB. 
We observe a severe degradation of the SIC ability of the multi-tap circuit due to the quantization errors induced by the phase shifters and attenuators.

%
To see the fundamental limits of achievable SIC performance with practical attenuators and phase shifters, we prove the following: 
\begin{corollary}
	Even if the power of noise $\sigma^2$ is zero, and the time delays of the multi-tap circuit perfectly match that of the SI channel, the SI still remains; this is due to the quantization errors induced by the phase shifters and attenuators, where the minimum average power of the effective SI channel with the multi-tap circuit is 
	\begin{align}
	\label{eq.cor2}
	\tilde{P}^q_{H_{\text{eff}}} = \frac{1-2PA_1+A_2}{K}{\mathrm{tr}}(\bold{E_{H_{{\mathrm{SI}}}H_{{\mathrm{SI}}}^*}}).	
	\end{align}
\end{corollary}
\begin{proof}
	Let $\tau_i=iT$. Then we get $\bold{R}=K\bold{I_M}$, since 
	\begin{align}
	R_{ab} =& \sum_{k=1}^{K}e^{-jk\Delta_wT(a-b)}\\ \nonumber
	=&\begin{cases}
	K &\text{for }a=b\\
	0 &\text{for } a\neq b.
	\end{cases}   
	\end{align}
	By substituting $R=K\bold{I_M}$, \eqref{eq.prop1} can be simplified as 
	\begin{align}
	\label{eq.cor3}
	\tilde{P}^q_{H_{\text{eff}}} =& \frac{1}{K}{\mathrm{tr}}(\bold{E_{H_{{\mathrm{SI}}}H_{{\mathrm{SI}}}^*}})-
\frac{2PA_1-A_2}{K}{\mathrm{tr}}(\bold{\Omega} \bold{R}^{-1} \bold{\Omega}^* \bold{E_{H_{{\mathrm{SI}}}H_{{\mathrm{SI}}}^*}})\\ \nonumber &+ \frac{\sigma^2}{K}A_2 M.
	\end{align}
	When the number of circuit's taps, $M$, is greater than or equal to the number of SI channel's taps,  $L_f$,  ${\mathrm{tr}}(\bold{\Omega} \bold{R}^{-1} \bold{\Omega}^* \bold{E_{H_{{\mathrm{SI}}}H_{{\mathrm{SI}}}^*}})$ becomes ${\mathrm{tr}}\left(\bold{E_{H_{{\mathrm{SI}}}H_{{\mathrm{SI}}}^*}}\right)$~\cite{theoretical}.
	With this, we get~\eqref{eq.cor2} by substituting $\sigma^2\!=\!0$ into \eqref{eq.cor3}, since we assume the SI channel estimation is perfect.
\end{proof}

Note that $K,\bold{E_{H_{{\mathrm{SI}}}H_{{\mathrm{SI}}}^*}}, \bold{\Omega}$, and $M$ are the parameters associated with the SI channel and the circuit's tap delay configuration. 
The term $(1\!\!-\!\!2PA_1\!\!+\!\!A_2)$ can be interpreted as an indicator of the SIC performance degradation due to the limited resolutions of the phase shifters and attenuators. If we assume the ideal phase shifters and attenuators,~\eqref{eq.cor2} goes to 0, which implies the perfect SIC.
The key factor of the analog SIC performance, as demonstrated by the simulation and analytical analysis, are the resolutions of the phase shifters and attenuators. Thus, the resolutions of the phase shifters and attenuators should be taken account into the design of the multi-tap circuit and it's optimization algorithm.

\subsection{Non-Ideal Power Amplifier}
\label{sec.pa}

\begin{figure}
	\begin{center}
		\resizebox{3.3in}{!}{\includegraphics{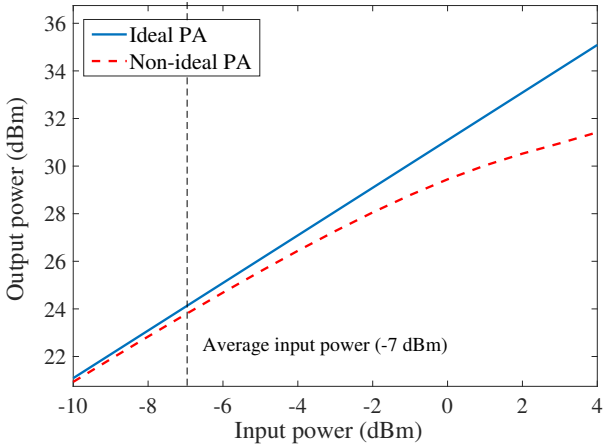}}
		\caption{A fitting curve of the input and output signals of the PA (red line). The corresponding polynomial is $x_{\text{PA}}=35.89x-2.24|x|^{2}x$. The practical PA adopted in~\cite{prototyping,minsoomag} and PXIe platform. }	
		\label{PA_input_vs_output}
	\end{center}
\end{figure}

In this section, we investigate the effect of PA nonlinearity on the analog SIC performance. In the systems with high peak-to-average power ratio (PAPR) such as orthogonal frequency division multiplexing (OFDM), researchers consider PA nonlinearity to be the main hurdle in the digital SIC~\cite{nonlinDSIC,iterative_digital_nonlin,measurements2015,unitygain}. 
As the PA nonlinearity deteriorates the accuracy of the SI channel estimation, it also affects the SIC performance of the multi-tap circuit. The authors in~\cite{MIMO} implemented an OFDM-based full-duplex system with the multi-tap circuit and observed that the transmitter produces the nonlinearities 30 dB lower than the transmitted signal. This phenomenon fundamentally limits the multi-tap circuit's SIC performance to 30 dB. Therefore, the authors in~\cite{MIMO} proposed an iterative tuning algorithm to solve it and achieved 60 dB cancellation. 

We adopt the parallel Hammerstein model to analyze a channel estimation error caused by the practical (i.e., nonlinear) PA.
The parallel Hammerstein model represents the relationship between the input and output signals of the PA as follows:
\begin{align}
\label{eq.pa}
x_{\text{PA}}[n]=&\sum_{p=0}^{P-1}\psi_{2p+1}|x[n]|^{2p}x[n]\nonumber\\
=&\psi_1x[n]+\underbrace{\psi_{3}|x[n]|^{2}x[n]}_{x_3[n]}+\cdots,
\end{align} 
where $x[n]$ and ${x_{\text{PA}}}[n]$ is the transmitted and power amplifier output signals on time $n$, $2P\!-\!1$ is the highest order of the model, and $\psi_{p}$ are the nonlinear coefficients.  
Note that the higher-order terms can usually ignored~\cite{minsoomag,MIMO}. 
We adopt the 3-order parallel Hammerstein model (i.e., $P=2$).
In~\eqref{eq.pa}, $\psi_1$ and $\psi_3$ denote the linear gain and 3-order gain in the time-domain signal, respectively.  
The authors in~\cite{minsoomag} fitted the parallel Hammerstein model according to the measured input/output power of the PA\footnote{Mini-Circuits ZVA-183W+ Super Ultra Wideband Amplifier \\http://www.minicircuits.com/pdfs/ZVA-183W+.pdf}.
We depict the power of input and output signals of the modeled PA in Fig.~\ref{PA_input_vs_output}, where the corresponding polynomial is $x_{\text{PA}}[n]=35.89x[n]-2.24|x[n]|^{2}x[n]$.
In the theoretical analysis, we assume that the linear gain $\psi_1$ is 1.
\begin{figure}
	\begin{center}
		\resizebox{3.5in}{!}{\includegraphics{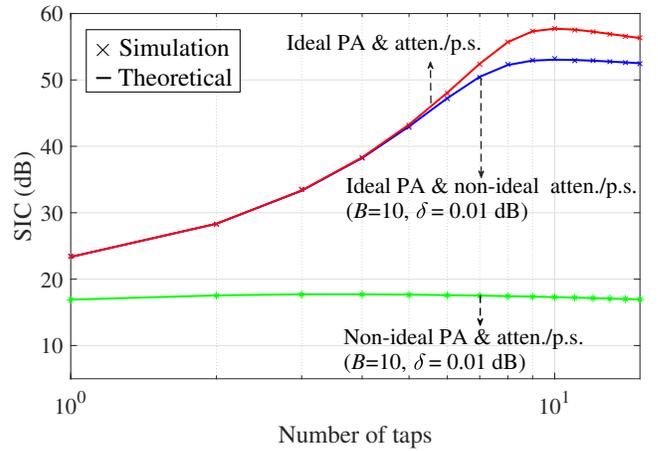}}
		\caption{SIC versus the number of taps ($M$) with the ideal PA (the  red and blue lines) and the practical PA (the green line).	  
			The estimation error of SI channel induced by noise, $r=-50~\text{dB}$. $B\!=\!10$ and $\delta\!=\!0.01~\text{dB}$.}	
		\label{PA_nonlin_MS_vs_lin}
	\end{center}
\end{figure}
 \begin{figure*}[!t]
 	\begin{center}	
 		\subfigure[]	{\includegraphics[width=0.90\columnwidth,keepaspectratio]%
 			{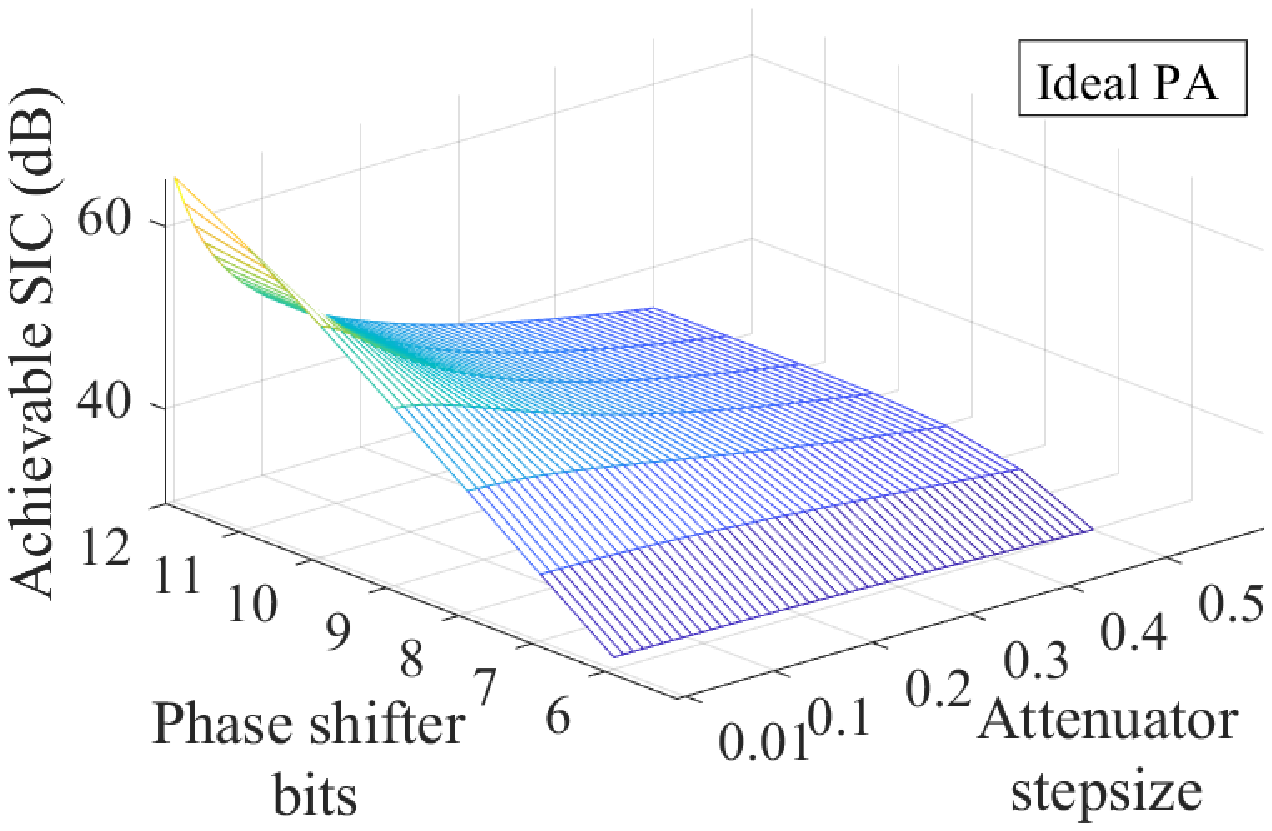}
 			\label{mesh_lin}}{\phantom{123}}
 		\subfigure[]
 		{\includegraphics[width=0.90\columnwidth,keepaspectratio]%
 			{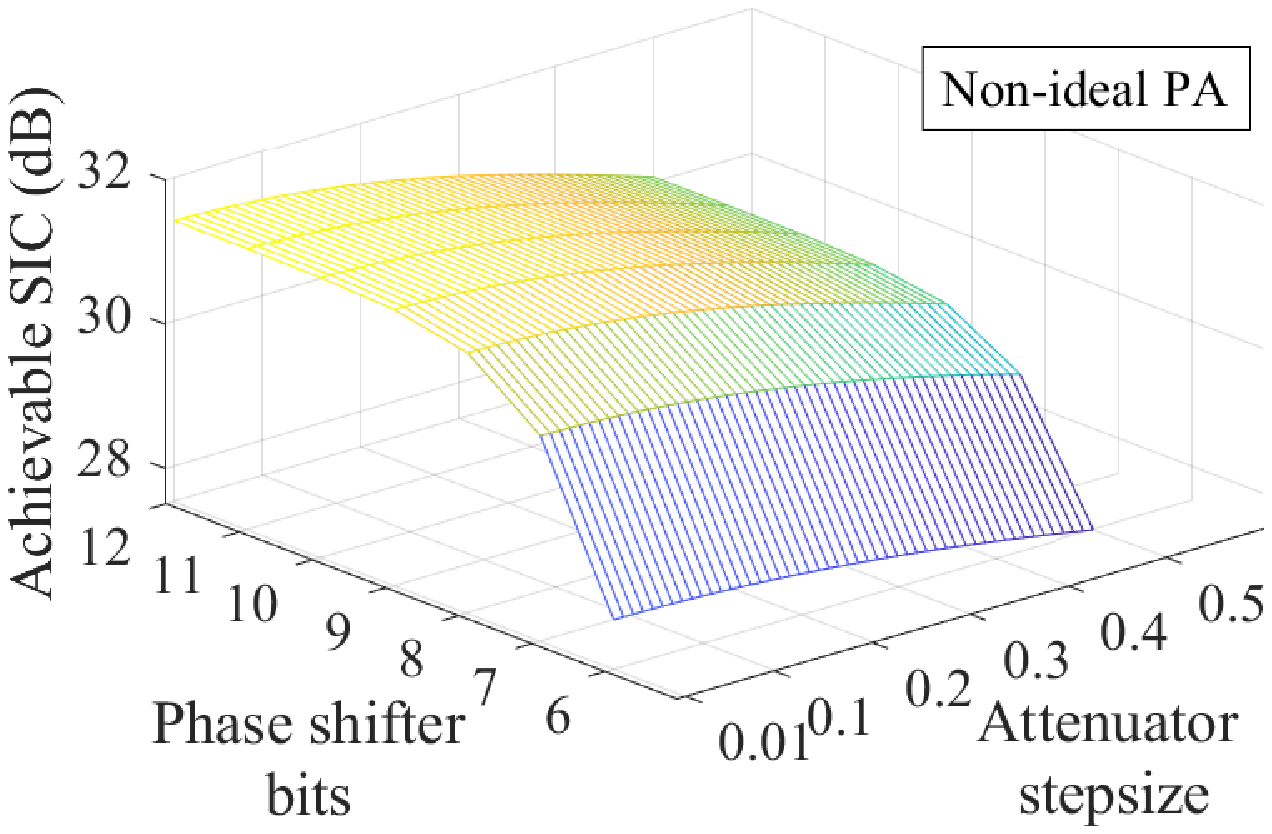}
 			\label{mesh_nonlin}}
 		\caption{The achievable SIC with different phase shifter and attenuator resolutions. (a) The case of the ideal PA. (b) The case of the non-ideal PA (see Fig.~\ref{PA_input_vs_output}).}
 	\end{center}
 \end{figure*}
For the third case, we employ the measured PA characteristics .     
As we assume unity linear gain in the theoretical analysis, we normalized the polynomial by $\psi_1=35.89$, which yields us $x_{\text{PA}}[n]=x[n]-0.06|x[n]|^{2}x[n]$.
respectively. The power of the PA's input signal is set at -7 dBm~\cite{minsoomag}.
Reflecting the nonlinear distortions in the SI channel estimation, \eqref{h_noise} can be rewritten as 
\begin{equation}
\label{h_noise2}
{\hat{H}_{\text{SI}}}[k] = {H_{\text{SI}}}[k]+\frac{{X_3}[k]{H_{\text{SI}}}[k]}{{X}[k]} +{N}[k],
\end{equation}
where ${X}[k]$ and ${X_3}[k]$ denote the frequency-domain gain of $\bold{x}$ and $\bold{x_3}$ for $k$-th subcarrier, respectively. The term $\frac{{X_3}[k]{H_{\text{SI}}}[k]}{{X}[k]}$ represents the channel estimation error induced by the PA nonlinearity. We derive the SIC performance with nonlinear PA by substituting ${N}[k]$ in~\eqref{eq.residual_q2} with $\frac{{X_3}[k]{H_{\text{SI}}}[k]}{{X}[k]}+{N}[k]$. For convenience, we define the following system parameters:
\begin{align}
\label{eq.powers}
p_1&\stackrel{\text{def}}{=} E\big[|{X}[k]|^2\big],\nonumber\\
p_2&\stackrel{\text{def}}{=} E\big[|{X}[k]|^4\big]\nonumber\\
p_3&\stackrel{\text{def}}{=} E\bigg[\frac{1}{|{X}[k]|^2}\bigg].
\end{align}

In the derivation of Theorem 1, we use the properties of $\bold{N}$, such that $\mathbb{E}[{N}[k]]=0$, and $\mathbb{E}[\bold{NN}^*]$ is a diagonal matrix where the diagonal elements are $\sigma^2$. These properties no longer hold after we consider the distortion induced by the PA. Instead, we utilized the following Lemma.
\begin{lem}
	Let ${N_{\text{PA}}}[k]\stackrel{\text{def}}{=}\frac{{X_3}[k]}{{X}[k]}$, then,
	\begin{align}
	\label{eq.lem_n2}
	m_1&\stackrel{\text{def}}{=}\mathbb{E}[{N_{\text{PA}}}[k]]=\psi_{3}p_1(2K-1)K, \nonumber\\ m_2\!&\stackrel{\text{def}}{=}\!\mathbb{E}[|{N_\text{PA}}[k]|^2], \nonumber\\
	&=\frac{\psi_{3}^2}{K^2}\left\{4(K\!-1)^2p_1+\!(4K\!-\!3)p_2+(K-1)p_2p_3\right.  \nonumber\\  m_3\!&\stackrel{\text{def}}{=}\!\mathbb{E}[{N_\text{PA}}[k_1]{N^*_\text{PA}}[k_2]], \ \  for \  k_1\neq k_2\nonumber \\
	&=\frac{\psi_{3}^2}{K^2}\left\{(4K^2-6K)p_1+4(K-1)p_2\right\}.
	\end{align}
\end{lem}
\begin{proof}
	See Appendix-B.
\end{proof}
Based on Lemma 1, we derive the following:
\begin{theorem}
	The average power of the effective SI channel with the multi-tap circuit with $B$-bit phase shifters, attenuator stepsize $\delta$, and nonlinear power amplifier with 3-order nonlinear coefficient $\psi_3$,  
	$P^{q,\text{PA}}_{H_{\text{eff}}}$, can be rewritten as
	\begin{align}
	\label{eq.prop2}
	&P^{q,\text{PA}}_{H_{\text{eff}}} = \frac{1}{K}{\mathrm{tr}}[\bold{E_{H_{{\mathrm{SI}}}H_{{\mathrm{SI}}}^*}}]\nonumber \\
	+&\frac{(1\!\!+\!\!2m_1\!+\!m_3)(PA_1)^2\!-\!2(m_1\!+\!1)\!PA_1}{K}{\mathrm{tr}}[\bold{\Omega} \bold{R}^{-1}\bold{\Omega}^*\bold{E_{H_{{\mathrm{SI}}}H_{{\mathrm{SI}}}^*}}] \nonumber\\ 
	+&\frac{\left(m_2\!-\!m_3\right)(PA_1)^2}{K}{\mathrm{tr}}\left[\bold{\Omega} \bold{R}^{-1}\bold{\Omega}^*\bold{D(E_{H_{\mathrm{SI}}H_{\mathrm{SI}}^*})} \right] \nonumber \\
	+&\frac{\sigma^2}{K} \left[(PA_1)^2M+(A_2-(PA_1)^2)K{\mathrm{tr}}(\bold{R}^{-1}) \right] \nonumber \\
	+&\left(1\!\!+\!\!2m_1\!+\!m_3)(A_2-(PA_1)^2\right){\mathrm{tr}}\left[\bold{R}^{-1}\bold{\Omega}^*\bold{E_{H_{\mathrm{SI}}H_{\mathrm{SI}}^*}}\bold{\Omega} \bold{R}^{-1} \right] \nonumber \\
	+&\left(m_2\!-\!m_3)(A_2-(PA_1)^2\right){\mathrm{tr}}\left[\bold{R}^{-1}\bold{\Omega}^*\bold{D(E_{H_{\mathrm{SI}}H_{\mathrm{SI}}^*})}\bold{\Omega} \bold{R}^{-1} \right], 
	\end{align}
	where $\bold{D(E_{H_{\mathrm{SI}}H_{\mathrm{SI}}^*})}$ denotes the diagonal matrix whose diagonal elements are equal to that of $\bold{E_{H_{\mathrm{SI}}H_{\mathrm{SI}}^*}}$. The quantities  $P,A_1,A_2,m_1,m_2$ and $m_3$ are defined in~\eqref{eq.prop1} and~\eqref{eq.lem_n2}.
\end{theorem}
\begin{proof}
	See Appendix-C.
\end{proof}	

\begin{figure*}[t]
	\centerline{\resizebox{2\columnwidth}{!}{\includegraphics{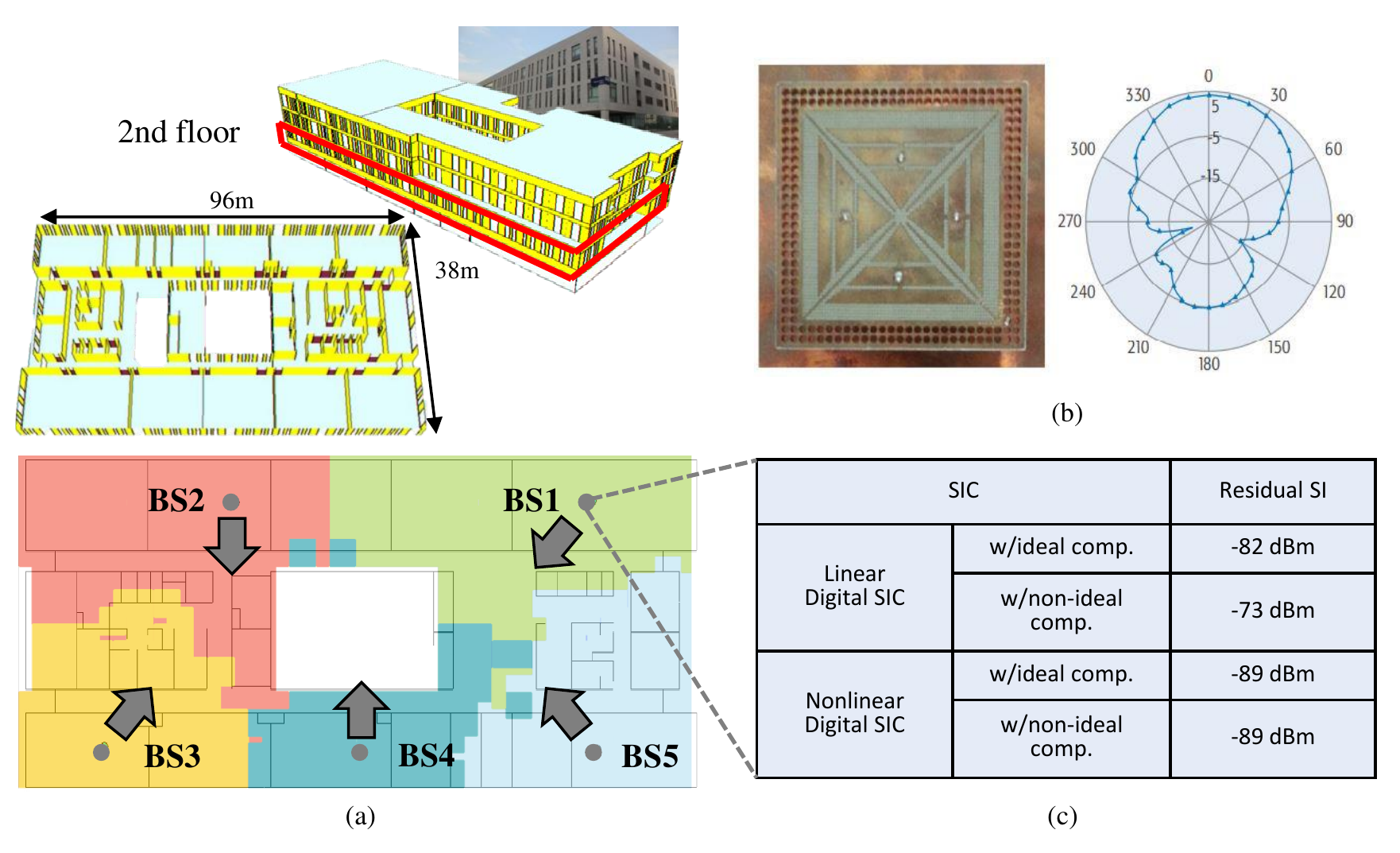}}}
	\caption{(a) The modeled building structure and the base station (BS) deployments for the system-level throughput analysis. (b) The measured radiation pattern of the adopted dual-polarized antenna. (c) The link-level SIC simulation results for BS1. Note that the noise floor is -90~dBm. }
	\label{SLS}
\end{figure*}
\begin{figure}[t]
	\centerline{\resizebox{\columnwidth}{!}{\includegraphics{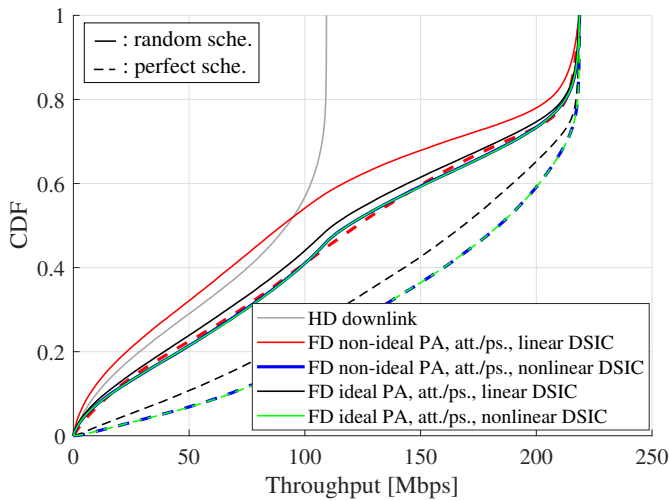}}}
	\caption{The CDFs of the system throughput values.}
	\label{cdf}
\end{figure}

In Lemma 1, $m_1,m_2$, and $m_3$ are the parameters associated with the nonlinear distortions induced by PA. Note that these parameters are expressed in terms of the 3-order gain of the PA, $\psi_3$, and the system parameters $p_1,p_2,p_3$, and $K$.   
Fig.~\ref{PA_nonlin_MS_vs_lin} shows a comparison between the SIC performance with the ideal and practical PAs.
We compare the following three scenarios: i) ideal PA, attenuators and phase shifters, ii) ideal PA, non-ideal attenuators and phase shifters ($\delta=0.01$ dB, $B=10$ bits), and iii) non-ideal PA, attenuators and phase shifters ($\delta=0.01$ dB, $B=10$ bits). 
We observe the notable differences between the SIC performances of the case ii) and iii) (i.e., the blue line and the green line, respectively). To analyze the impact of the non-idealities of the PA and the phase shifters/attenuators separately, we derive the following corollary: 
\begin{corollary}
	Even if the power of noise $\sigma^2$ is zero, and the time delays of the multi-tap circuit perfectly match that of the SI channel, the SI still remains due to the PA's nonlinear distortions and the quantization errors induced by the phase shifters and attenuators, where the minimum average power of the effective SI channel with the multi-tap circuit is 
	\begin{align}
	\label{eq.cor3}
		\tilde{P}^{q,\text{PA}}_{H_{\text{eff}}} &= \tilde{P}^{q}_{H_{\text{eff}}} -
		{\mathrm{tr}}( \bold{E_{H_{{\mathrm{SI}}}H_{{\mathrm{SI}}}^*}})\!\times\!  \Big(\!\frac{2m_1(PA_1\!-\!A_2)\!-\!m_3A_2}{K^2}\!\Big) \nonumber \\
	&+ {\mathrm{tr}}(\bold{\Omega} \bold{\Omega}^* \bold{D(E_{H_{\mathrm{SI}}H_{\mathrm{SI}}^*})})\!\times\!
	\Big(\!\frac{A_2(m_2-m_3)}{K^2}\!\Big).
	\end{align}
\end{corollary}
\begin{proof}
	See Appendix-D.
\end{proof}
In~\eqref{eq.cor3}, we present the minimum average power of the effective SI channel with the non-ideal PA, phase shifters, and attenuators, $\tilde{P}^{q,\text{PA}}_{H_{\text{eff}}}$. Note that $\tilde{P}^{q,\text{PA}}_{H_{\text{eff}}}$ consists of $\tilde{P}^{q}_{H_{\text{eff}}}$, and the terms associated with the nonlinear distortions induced by PA.    
Based on~\eqref{eq.cor2} and~\eqref{eq.cor3}, we depict the achievable SIC of the multi-tap circuit with ideal and non-ideal PA in Fig.~\ref{mesh_lin} and Fig.~\ref{mesh_nonlin}, respectively.

\begin{table}
	\caption{Simulation parameters}
	\label{table_1}
	\begin{center}
		\begin{tabular}{ccc}
			\hline\hline
			System Parameter &Notation &Values \\
			\hline
			Center frequency & &2.52GHz\\
			Bandwidth & &20MHz\\		
			Modulation & &64QAM\\
			FFT size &&2048\\
			Used subcarrier &&1200\\
			CP length&&512\\
			Tx power&&23dBm\\
			Power amplifier gain && 30dB\\		
			Received noise floor&&-90dBm\\
			Phase shifter bits &$B$&8bit\\
			Attenuator step size &$\delta$& 0.1dB\\
			ADC bits &&14bit\\
			\hline \hline
		\end{tabular}
	\end{center}
\end{table}

\section{System Level Throughput Analysis}
\label{sec_4}
In this section, we provide system-level throughput evaluations of a full-duplex system equipped with analog SIC via multi-tap circuit in indoor multi-cell environments.
In Fig.~\ref{SLS}(a), we illustrate the modeled building structure and BS deployments. We deployed the five BSs equipped with dual-polarized antennas~\cite{prototyping,dualpole}. The adopted antenna  has high cross-polarization discrimination (XPD) characteristic and it achieves 40 dB of isolation.
For the indoor channel modeling, we utilized the 3D ray-tracing tool developed by Bell Labs~\cite{3d-smallcell,3d-vtc,yg}, Wireless System Engineering (WiSE). The measured radiation pattern of the dual-polarized antenna was employed in the 3D ray-tracing. 
We calculated the power-delay-profile (PDP)
of the SI channel for each BS through 3D ray-tracing and reflected it in the link-level SIC simulations. Through 3D ray-tracing, we observed that the time delays of the SI channel of BS1 are distributed in the range of 5 ns to 40 ns, where the gain of the direct path is about -40 dB. We adopt the multi-tap circuit which has 4-taps with the time delays evenly spaced in the range of 5 ns to 40 ns. In Table 1, we present the simulation parameters for the performance evalutations.

\begin{figure*}[t]
	\centerline{\resizebox{2\columnwidth}{!}{\includegraphics{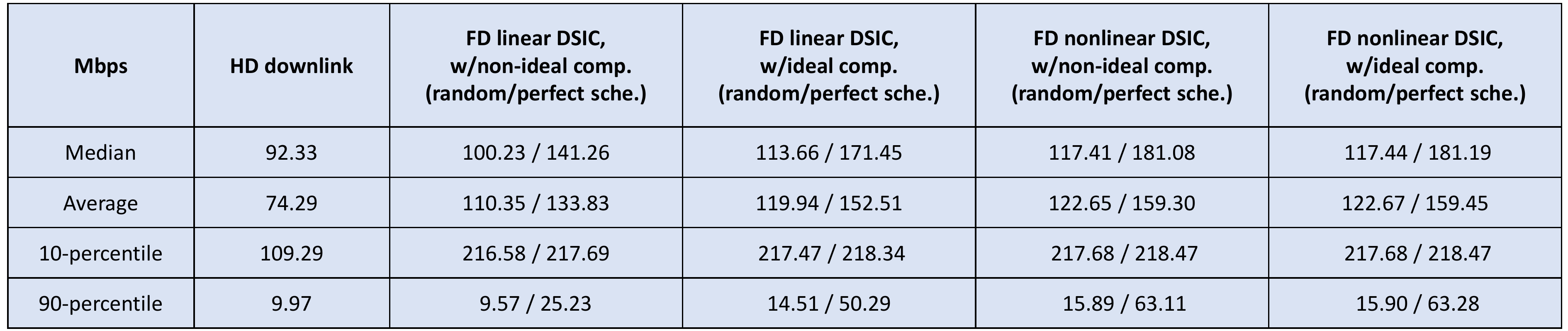}}}
	\caption{The representative system throughput values of the half-duplex system and the full-duplex system with different SIC scenarios.}
	\label{throughput_values}
\end{figure*}
In Fig.~\ref{SLS}(c), we present the link-level SIC results for BS 1. We adopt the two different SIC scenarios for both the analog and digital SIC. For the analog SIC, we consider the ideal and non-ideal RF components (i.e., 8-bit phase shifter, the attenuator with stepsize 0.1~dB, and the PA characteristics introduced in~\cite{minsoomag}) in the simulations.  
The residual SI is once again mitigated by the linear/nonlinear digital SIC algorithms. 
The linear digital SIC algorithm cancels the residual SI in the frequency domain by the least-square method.  
The nonlinear digital SIC algorithm herein refers to the algorithm proposed in~\cite{nonlinDSIC}, which estimates and cancels the SI based on the parallel Hammerstein model. Based on~\eqref{eq.pa}, the received signal is modeled as
\begin{align}
	\label{eq.pa1}
	y[n]&=\sum_{\ell=0}^{L-1}{h_{\text{SI}}}[\ell]x_{\text{PA}}[n-\ell]+z[n] ] \nonumber   \\
	&=\sum_{p=0}^{P-1} \ \sum_{\ell=0}^{L-1}b_{p,\ell}f_{p}(x[n-\ell])+z[n],
\end{align} 
where $y[n]$ and $z[n]$ are the received signal and noise at time~$n$, ${h_{\text{SI}}}$ is the residual SI channel after analog SIC, $L$ is the number of taps of the residual SI channel, $\{b_{p,\ell}\}$ are the effective nonlinear coefficients and $f_{p}(x[n])=|x[n]|^{2p}x[n]$.
We estiamte the effective nonlinear coefficients, $\{b_{p,\ell}\}$, by transmitting a preamble. The highest order of the parallel Hammerstein model ($2P\!-\!1$), and the number of taps of the residual SI channel ($L$) are set to 3 and 8, respectively. Let $S$ denotes the number of observed samples, then,
the least-square estimation of the effective nonlinear coefficients, $\pmb{\hat{b}}$, can be obtained as follows:
\begin{align}
	\pmb{\hat{b}}&=\underset{\pmb{b}}{\argmin}||\pmb{y}-\pmb{f}\pmb{b}||^2 \nonumber\\
	&=(\pmb{f}^{*}\pmb{f})^{-1}\pmb{f}^{*}\pmb{y}.
\end{align}
where
\begin{gather}
	\pmb{b}=[b_{0,0} \ b_{0,1} .. b_{0,L-1} .. b_{P\!-\!1,0} \ b_{P\!-\!1,1} .. b_{P\!-\!1,L-1}],\nonumber \\
	\pmb{f}=[\pmb{f}_{0} \ \pmb{f}_{1} \ ..\pmb{f}_{P-1}],  \nonumber\\
	\pmb{f}_{p}=\setlength\arraycolsep{0.8pt}\begin{bmatrix}f_{p}({x}[n])&f_{p}({x}[n-1])&\cdots&f_{p}({x}[n-\!L+1])\\f_{p}({x}[n+1])&f_{p}({x}[n])&\cdots&f_{p}({x}[n\!-\!L\!+\!2])\\\vdots&\vdots&\ddots&\vdots\\f_{p}({x}[n\!+\!S\!-\!1])&f_{p}({X}[n\!+\!S\!-\!2])&\cdots&f_{p}({x}[n\!-\!L\!+\!S])\end{bmatrix}.  
\end{gather}

With the linear digital SIC, the total SIC performance degraded by 9~dB taking into account the non-idealities of the RF components (i.e., the residual SI power increases -82~dBm to -73 dBm). On the other hand, the nonlinear digital SIC algorithm cancels the SI close to the noise floor regardless the analog SIC performance degradation due to the non-idealities of the RF components.
This happens because the linear digital SIC cannot suppress the nonlinear distortion which mainly comes from the PA nonlinearity. The cancellation of the nonlinear distortion is carried by the analog SIC and the nonlinear digital SIC. 
The multi-tap circuit is able to cancel the nonlinear distortions since the circuit takes the PA's output as an input. 
With the practical RF components the necessity of the nonlinear digital SIC is emphasized as it can make up the cancellation of the nonlinear distortions.
Note that even if the RF components are assumed to be ideal, the nonlinear digital SIC shows better performance than the linear digital SIC. It is because the linear digital SIC algorithm estimates $K$ channel coefficients in the frequency domain, while the nonlinear digital SIC algorithm estimates only $L$ effective nonlinear coefficients in the time domain. The comparison of the accuracies of the time domain and frequency domain least-square estimates of the channel response in OFDM system is presented in~\cite{time_vs_freq}. 

Fig.~\ref{cdf} depicts the ergodic throughputs of the half-duplex and full-duplex system with the four different SIC scenarios.
For each SIC scenario, the solid line in Fig.~\ref{cdf} corresponds to the cumulative distribution function (CDF) of the system throughputs with random user equipment (UE) selection. The dashed lines are correspond to the CDF of the system throughputs without UE-to-UE interference, which can be achieved by the perfect UE selection. 
Note that the upper bounds of the system throughputs are indicated by no UE-to-UE interference cases. 
  
With the non-ideal RF components and linear digital SIC (the solid red line), the average throughput is lower than that of the half-duplex system when UE-to-UE interference exists.
The gap between the red and black lines can be interpreted as the impacts of residual nonlinear SI components, which cannot be mitigated by the linear digital SIC. 
With the nonlinear digitla SIC (i.e., the blue and green lines), the SI is suppressed close to the receiver noise floor. 

The system throughput values for the half-duplex system and the full-duplex system with different SIC scenarios are represented in Fig.~\ref{throughput_values}.
Considering UE-to-UE interference, the average throughputs of the full duplex system with ideal/non-ideal RF components and linear digital SIC (the solid red line and the solid blue line, respectively) increased by 61 and 48 percent over that of the half duplex system, respectively. 
The overall results make clear that the non-idealities of the RF components significantly affect the analog SIC performance, which puts a burden on the digital SIC.

\section{{\fontsize{11}{14}\selectfont Concluding Remarks}}
\label{Sec.conclusion}

In this paper, we investigated the analog SIC performance of the multi-tap circuit with practical RF components. We theoretically derived the SIC performance considering the resolution of the attenuators and phase shifters and the PA nonlinearity. We have verified the derived formula through link-level SIC simulations.
For the realistic performance analysis, we employed measured PA characteristics. We carried system-level throughput evaluations in an interference-limited environment via 3D ray-tracing.
Our results manifest that we should consider the limited resolution of the attenuators and phase shifters in the design of analog SIC.

\appendices

\section*{Acknowledgment}

\section{Proof of Theorem 1}
Using the facts that i) $\text{tr}(\cdot)$ and $\mathbb{E}$ are commutative and ii) $\text{tr}(AB)=\text{tr}(BA)$, we can expand \eqref{eq.residual_q} as 
\begin{align}
\label{eq.residual_q2}
	P^q_{H_{\text{eff}}} &= \frac{1}{K}\text{tr}(\mathbb{E}[\bold{H_\text{SI}}\bold{H_\text{SI}}^*]) -\underbrace{\frac{2}{K}\text{tr}(\mathbb{E}[\bold{\Omega} \bold{Q}\bold{R}^{-1}\bold{\Omega}^* \bold{H_{\text{SI}}}\bold{H_{\text{SI}}}^*])}_{(a)}\nonumber \\
&+\underbrace{\frac{1}{K}\text{tr}(\mathbb{E}[\bold{\Omega} \bold{R}^{-1}\bold{Q}^*\bold{R}\bold{Q}\bold{R}^{-1}\bold{\Omega}^*\bold{H_{\text{SI}}}\bold{H_{\text{SI}}}^*])}_{(b)} \nonumber \\
&+\underbrace{\frac{1}{K}\text{tr}(\mathbb{E}[\bold{\Omega} \bold{R}^{-1}\bold{Q}^*\bold{R}\bold{Q}\bold{R}^{-1}\bold{\Omega}^*\bold{NN^*}])}_{(c)} \nonumber \\
&+\underbrace{\frac{1}{K}\text{tr}(\mathbb{E}[\bold{\Omega} \bold{R}^{-1}\bold{Q}^*\bold{R}\bold{Q}\bold{R}^{-1}\bold{\Omega}^*(\bold{NH^*_\text{SI}+\bold{H_\text{SI}}N^*)}])}_{(d)} \nonumber \\
&-\underbrace{\frac{2}{K}\text{tr}(\mathbb{E}[\bold{\Omega} \bold{Q}\bold{R}^{-1}\bold{\Omega}^* \bold{N}\bold{H_{\text{SI}}}^*])}_{(e)}.
\end{align} 
Since $\mathbb{E}[{N}[k]]=0$, the last two terms, $(d,e)$ are zero.
Note that the elements in $\bold{N}$ and $\bold{Q}$ are independent, and all the other matrices in the three expectation terms in~\eqref{eq.residual_q2} are deterministic. Therefore, the term $(a)$ in~\eqref{eq.residual_q2} is equal to 
\begin{align}
\label{eq.term_a}
(a) &= \frac{2\mathbb{E}[{Q}_{ii}]}{K}\text{tr}(\bold{\Omega} \bold{R}^{-1}\bold{\Omega}^* \mathbb{E}[\bold{H_{\text{SI}}}\bold{H_{\text{SI}}}^*]). 
\end{align} 
To simplify the terms $(b)$ and $(c)$, we first derive the following property. 
For the deterministic matrices $\bold{A}$, $\bold{B}$ and the quantization matrix $\bold{Q}$,
\begin{align}
\label{eq.lemmaprof}
&\text{tr}\left(\mathbb{E}[\bold{Q^*AQB}]\right)\nonumber\\ &=\mathbb{E}\left[\sum_{i=1}^{M}\sum_{j=1}^{M}\bold{Q}^*_{ii}{Q}_{jj}\bold{A}_{ij}\bold{B}_{ji}\right] \nonumber \\
&=\mathbb{E}\left[\sum_{i=1}^{M}|{Q}_{ii}|^2{A}_{ii}{B}_{ii}+\sum_{i\neq j, i=1}^{M}\sum_{j=1}^{M}\bold{Q}^*_{ii}{Q}_{jj}\bold{A}_{ij}\bold{B}_{ji}\right] \nonumber \\
&=\mathbb{E}[|{Q}_{ii}|^2]\!\!\left(\sum_{i=1}^{M}{A}_{ii}{B}_{ii}\right)+\mathbb{E}[\bold{Q}^*_{ii}{Q}_{jj}]\!\!\left(\sum_{i\neq j, i=1}^{M}\sum_{j=1}^{M}\bold{A}_{ij}\bold{B}_{ji}\right) \nonumber \\
&=\left(\mathbb{E}[|{Q}_{ii}|^2]\!\!-\!\!\mathbb{E}[{Q}_{ii}]^2\right)\!\!\left(\sum_{i=1}^{M}{A}_{ii}{B}_{ii}\right)\!+\mathbb{E}[{Q}_{ii}]^2\mathrm{tr}(\bold{AB}). 
\end{align}

Note that $\mathbb{ E}[|{Q}_{ii}|^2]$ and $\mathbb{E}[{Q}_{ii}]$ can be calculated as  
\begin{align}
\mathbb{E}[{Q}_{ii}]&= \mathbb{E}[10^{(n_{a,i})/20}]\mathbb{E}[e^{-jn_{p,i}}]\nonumber\\
&=\frac{20}{\delta \ln(10)} (10^{\frac{\delta}{40}} - 10^{\frac{-\delta}{40}})\int_{-\frac{\pi}{2^{B}}}^{\frac{\pi}{2^{B}}}\frac{2^B}{2\pi}e^{-jx}dx \nonumber \\
&=\underbrace{\frac{20}{\delta \ln(10)} (10^{\frac{\delta}{40}} - 10^{\frac{-\delta}{40}})}_{A_1}\underbrace{\frac{2^B}{\pi} \sin(\frac{\pi}{2^B})}_{P}, \\
\mathbb{E}[|{Q}_{ii}|^2]&=\mathbb{E}[10^{n_{a,i}/10}] \nonumber\\ 
&=\int_{-\delta/2}^{\delta/2}\frac{1}{\delta}10^{x/10}dx \nonumber\\
&=\underbrace{\frac{10}{\delta \ln(10)} (10^{\frac{\delta}{20}} - 10^{\frac{-\delta}{20}})}_{A_2}.
\end{align}
In (26) and (27), we define three quantities, $(A_1, A_2) \ \text{and} \ P$, which are related to the attenuator stepsize and phase shifter bits, respectively. We can rewrite~\eqref{eq.lemmaprof} as 
\begin{align}
\label{eq.lemma2}
&\text{tr}(\mathbb{E}[\bold{Q^*AQB}]) \!\!=\!\!\left(A_2\!\!-\!\!P^2A_1^2\right)\!\!\left(\sum_{i=1}^{M}{A}_{ii}{B}_{ii}\right)\!+P^2A_1^2\mathrm{tr}(\bold{AB}). 
\end{align}
Using this property and the fact that $R_{ii}=K$, we obtain the followings:
\begin{align}
(a) &=\frac{2PA_1}{K}\text{tr}(\bold{\Omega} \bold{R}^{-1}\bold{\Omega}^* \bold{E_{H_{\mathrm{SI}}H_{\mathrm{SI}}^*}}), \nonumber\\
(b) &= \frac{1}{K}[\text{tr}(\mathbb{E}[\bold{Q^*RQ}(\bold{R}^{-1}\bold{\Omega}^*\bold{H_{\text{SI}}}\bold{H_{\text{SI}}}^*\bold{\Omega} \bold{R}^{-1})])] \nonumber\\
&=\frac{1}{K}\left[P^2A_1^2\text{tr}(\bold{\Omega} \bold{R}^{-1}\bold{\Omega}^*\bold{E_{H_{\mathrm{SI}}H_{\mathrm{SI}}^*}} )\right.\nonumber \\
& \ \left.+(A_2-P^2A_2^2)\text{tr}(\bold{\Omega} (\bold{R}^{-1})^2\bold{\Omega}^*\bold{E_{H_{\mathrm{SI}}H_{\mathrm{SI}}^*}})\right], \nonumber\\
(c) &=\frac{\sigma^2}{K}\left[\text{tr}(\mathbb{E}[\bold{Q^*RQ}\bold{R}^{-1}])\right] \nonumber \\
&=\frac{\sigma^2}{K}[P^2A_1^2M+K(A_2-P^2A_1^2){\mathrm{tr}}(\bold{R}^{-1})]. 
\end{align} 
We get~\eqref{eq.prop1} by integrating $(a),(b)$ and $(c)$.
\section{Proof of Lemma 1}
Since $x _3[n]$ is $\psi_3{X}[n]{X}^*[n]{X}[n]$, we can rewrite $\mathbb{E}[{N_{\text{PA}}}[k_1]]$ as follows, by using the properties of discrete Fourier transform (DFT).  
\begin{align}
\label{X_3}
&\mathbb{E}[{N_{\text{PA}}}[k_1]]\\ \nonumber
\!&=\!\mathbb{E}[\frac{\psi_3}{{X}[k_1]K}\!\sum_{m=0}^{K-1}\!\{\!\sum_{n=0}^{K-1}{X}[n]{X}^*[(n\!-\!m)_K]{X}[(k_1\!-\!m)_K]\}],
\end{align}
where $(\cdot)_K$ denotes the modulo-$K$ operation. Note that~\eqref{X_3} is the expectation of the sum of $K^2$ different terms (i.e., for the case of $(m,\! n)\!=\!(1,\!2)$, the corresponding term is $\frac{\psi_3{X}[2]{X}^*[1]{X}[k_1-1]}{{X}[k_1]K^2}$). We denote the corresponding terms for index $(m,n)$ as $P_{k_1}\!(m,n)$.
\begin{equation}
\label{eq.pmn}
P_{k_1}\!(m,n)=\frac{\psi_3{X}[n]{X}^*[(n\!-\!m)_K]{X}[(k_1\!-\!m)_K]}{{X}[k_1]K}.
\end{equation} 
To calculate~\eqref{X_3}, we classify the index $(m,n)$ as follows:
\begin{align}
\label{eq.subset}
&S_1^{k_1}=\{(m,n)|m=0,n=k_1\}, \ |S_1^{k_1}|=1,\nonumber \\
&S_2^{k_1}=\{(m,n)|m=0,n\neq k_1\}, |S_2^{k_1}|=K\!-\!1,\nonumber \\
&S_3^{k_1}=\{(m,n)|m\neq 0, n=k_1\}, |S_3^{k_1}|=K\!-\!1,\nonumber \\
&S_4^{k_1}=\{(m,n)|m\neq 0, n\!=\!(k_1\!-\!m)_K\}, |S_4^{k_1}|=K\!-\!1,\nonumber \\
&S_5^{k_1}=\{(m,n)|(m,n)\notin \{S_1^{k_1}\cup S_2^{k_1} \cup S_3^{k_1}\cup S_4^{k_1}\}\}, 
\end{align}
where $|S_1^{k_1}|\!\!=\!\!1, |S_2^{k_1}|\!\!=\!\!|S_3^{k_1}|\!\!=\!\!|S_4^{k_1}|\!\!=\!\!K\!-\!1$, and $|S_5^{k_1}|\!=\!K^2\!-\!3K\!-\!2$.
For the each subset, we calculate $\mathbb{E}[P_{k_1}\!(m,n)]$ as
\begin{align}
&\mathbb{E}[P_{k_1}\!(m,n)]\!\!=\!\!\begin{cases} &\!\!\!\!\!\!\mathbb{E}[\frac{\psi_3|{X}[k_1]|^2}{K}]=\frac{\psi_3 p_1}{K},  (\!m,n\!)\!\!\in\!\! S_1^{k_1}\\
&\!\!\!\!\!\!\mathbb{E}[\frac{\psi_3|{X}[n]|^2}{K^2}]=\frac{\psi_3 p_1}{K},  (\!m,n\!)\!\!\in\!\! S_2^{k_1}\\
&\!\!\!\!\!\!\mathbb{E}[\frac{\psi_3|{X}[(k_1\!-\!m)_K]|^2}{K^2}]=\frac{\psi_3 p_1}{K}, (\!m,n\!)\!\!\in\!\! S_3^{k_1}	\\
&\!\!\!\!\!\!\mathbb{E}[\frac{\psi_3{X}[(k_1\!-\!m)_K]^2\!{X}^*[(k_1\!-\!2m)_K]}{{X}[k_1]K}]\!\!=\!\!0,	(\!m,n\!)\!\!\in\!\! S_4^{k_1}\\
&\!\!\!\!\!\!\mathbb{E}[\frac{\psi_3\!{X}[n]\!{X}^*[(n\!-\!m)_K]\!{X}[(k_1\!-\!m)_K]}{{X}[k_1]K}]\!\!=\!\!0, (\!m,n\!)\!\!\in\!\! S_5^{k_1}.
\end{cases}
\end{align}
Note that the real and imaginary part of ${X}[k]$ are independent discrete-uniform random variable with mean and variance are $0$ and $1/2$. We get~\eqref{eq.lem_n2} by tally up all the cases listed in~\eqref{eq.subset}.
\begin{align}
\label{eq.subsetsum}
\mathbb{E}[{N_{\text{PA}}}[k_1]]&=\sum_{i=1}^{5}|S_i|\mathbb{E}[P_{k_1}\!(m,n)|(m,n)\in S_i]\nonumber \\
&=\frac{\psi_3p_1}{K}+\frac{\psi_3p_1(K-1)}{K}+\frac{\psi_3p_1(K-1)}{K} \nonumber \\
&=\frac{\psi_3p_1(2K-1)}{K}. 
\end{align}

From~\eqref{X_3}, we get 
\begin{align}
\label{N_k1k2}
E[{N_\text{PA}}[k_1]&{N_\text{PA}}^*[k_2]]=E[\frac{\psi_3^2}{{X}[k_1]{X}^*[k_2]K^2}\\ \nonumber
\times\sum_{m_1=0}^{K-1}\{\sum_{n_1=0}^{K-1}{X}[_1]&{X}^*[(n_1-m_1)_K]{X}[(k_1-m_1)_K]\} \\ \nonumber
\times\sum_{m_2=0}^{K-1}\{\sum_{n_2=0}^{K-1}{X}^*[n_2]&{X}[(n_2-m_2)_K]X^*[(k_2-m_2)_K]\}].
\end{align} 

We can interpret~\eqref{N_k1k2} as an expectation of the summation of the $K^4$ combinations of $P_{k_1}\!(m_1,n_1)$, and $P_{k_2}\!(m_2,n_2)$. 
Let ${{\mathrm{I_2}}}$, and $\mathrm{I_2}$ denote the indices $(m_1\!,n_1)$, and $(m_2\!,n_2)$, respectively. The concatenation of the two indices, ${{\mathrm{I_2}}}$, and $\mathrm{I_2}$, is denoted by ${\mathrm{I}}$.
Similar to~\eqref{eq.subset}, we classify the $K^4$ different terms as follows:  
\begin{align}
\label{eq.subset_n1_diag}
&\text{When} \ \ k_1=k_2, \nonumber \\  
&T_1\!=\!\{{\mathrm{I}}|{{\mathrm{I_2}}},\mathrm{I_2}\in S_1^{k_1}\!\}, \nonumber \\
&T_2\!=\!\{{\mathrm{I}}|{{\mathrm{I_2}}}\in S_1^{k_1}, \mathrm{I_2}\in \{S_2^{k_1}\cup S_3^{k_1}\}\}, \nonumber \\
&T_3\!=\!\{{\mathrm{I}}|{{\mathrm{I_2}}}\in \{S_2^{k_1}\cup S_3\}, \mathrm{I_2}\in S_1^{k_1}\}, \nonumber \\
&T_4\!=\!\{{\mathrm{I}}|{{\mathrm{I_2}}}\in S_2^{k_1}, \mathrm{I_2}={{\mathrm{I_2}}}\!\}, \nonumber \\
&T_5\!=\!\{{\mathrm{I}}|{{\mathrm{I_2}}}\in S_2^{k_1}\!, \mathrm{I_2}\!\in\!\! S_3^{k_1}\!,m_2\!\equiv\!k_1\!-\!n_1\!\}, \nonumber \\
&T_6\!=\!\{{\mathrm{I}}|\mathrm{I_2}\in S_2^{k_1},{{\mathrm{I_2}}}\!\!\in\!\! S_3^{k_1}\!,m_1\!\equiv\!k_1\!-\!n_2\!\}, \nonumber \\
&T_7\!=\!\{{\mathrm{I}}|{{\mathrm{I_2}}}\!\in S_3^{k_1}, \mathrm{I_2}\!\in S_3^{k_1}, m_2\!=\!m_1\}, \nonumber \\
&T_8\!=\!\{{\mathrm{I}}|\!{{\mathrm{I_2}}},\mathrm{I_2}\!\in \!\{\!S_2^{k_1}\cup S_3\!\}\}\!-\!\!\!\bigcup_{i\in\{4,5,6,7\}}\!\!\!\!\!\!T_i, \nonumber \\
&T_9\!=\!\{{\mathrm{I}}|m_1\!\neq\!0,n_1\!=\!n_2\!\equiv\!k_1\!-\!m_1\!, m_2\!=\!m_1\!\}, \nonumber \\
&T_{10}\!=\!\{{\mathrm{I}}|{{\mathrm{I_2}}}=\mathrm{I_2}\in S_5^{k_1},  \}, \nonumber \\
&T_{11}\!=\!\{{\mathrm{I}}|{{\mathrm{I_2}}}\in S_5^{k_1}\!,  m_2\!\equiv\!k_1\!-\!n_2,n_2\!\equiv\!k_1\!-\!m_2\! \},
\end{align}
where
\begin{align}
&\mathbb{E}[P_{k_1}\!({{\mathrm{I_2}}})P_{k_2}\!^*(\mathrm{I_2})]\\
&=\!\!\begin{cases} \frac{\psi_3^2 p_2}{K^2},   \ \ \ \ {\mathrm{I}}\in T_1\\
\frac{\psi_3^2 p_1^2}{K^2},  \ \ \ \   {\mathrm{I}} \in  \{T_2\cup T_3\cup T_8\}\\
\frac{\psi_3^2 p_2}{K^2},  \ \ \ \ {\mathrm{I}} \in  \{T_4\cup T_5\cup T_6\cup T_7\}\\
\frac{\psi_3^2 p_1p_2p_3}{K^2},  \ \ \ \ {\mathrm{I}} \in  \{T_9\}\\
\frac{\psi_3^2 p_1^3p_3}{K^2}, \ \ \ \  {\mathrm{I}}\in \{T_{10}\cup T_{11}\}. \\
\end{cases}
\end{align}
Then we get $m_2$.
\begin{align}
\label{eq.subsetsum2}
m_2=&\mathbb{E}[{N_{\text{PA}}}[k_1]{N^*_{\text{PA}}}[k_1]] \nonumber\\
=&\sum_{i=1}^{11}|T_i|\mathbb{E}[P_{k_1}\!({{\mathrm{I_2}}})P_{k_2}^*\!(\mathrm{I_2})|{\mathrm{I}}\in T_i]\nonumber \\
=&\frac{\psi_{3}^2}{K^2}\left\{(4K^2\!-\!10K\!+\!6\!)p_1^2+\!(4K\!-\!3)p_2+(\!K\!-\!1\!)p_1p_2p_3\right.  \nonumber\\
&\left.+ 2(K^2\!\!-\!\!3K\!+\!2)p_1^3p_3\right\}. 
\end{align}
We derive (29) in the same way.
\begin{align}
\label{eq.subset_n1_nondiag}
&\text{When} \ \ k_1\neq k_2, \nonumber \\  
&U_1\!=\!\{{\mathrm{I}}|{{\mathrm{I_2}}}\in S_1^{k_1},\mathrm{I_2}\in S_1^{k_2}\}, \nonumber \\
&U_2\!=\!\{{\mathrm{I}}|{{\mathrm{I_2}}}\!\in \!S_1^{k_1}\!,m_2\equiv\!n_2\!-\!k_1,n_2\!\in\!\{k_1,k_2\}\!\}, \nonumber \\
&U_3\!=\!\{{\mathrm{I}}|m_1\!\equiv\!n_1\!-\!k_2,n_1\!\!\in\!\!\{k_1,k_2\}\!\},\mathrm{I_2}\in\! S_1^{k_2}\!\}, \nonumber \\
&U_4\!=\!\{{\mathrm{I}}|{{\mathrm{I_2}}}\!\!\in\!\! S_1^{k_1}\!, \mathrm{I_2}\in\!\!\{S_2^{k_2}\cup S_3^{k_2}\}\!\}\!-\!U_2, \nonumber \\
&U_5\!=\!\{{\mathrm{I}}|{{\mathrm{I_2}}}\!\!\in\!\!\{S_2^{k_1}\cup S_3^{k_1}\},\mathrm{I_2}\in\!\! S_1^{k_2} \!\}\!-\!U_3, \nonumber \\
&U_6\!=\!\{{\mathrm{I}}|{{\mathrm{I_2}}}\!\!\in\!\! S_2^{k_1},\mathrm{I_2}\in S_2^{k_2}, n_1\!=\!n_2 \!\}, \nonumber \\
&U_7\!=\!\{{\mathrm{I}}|{{\mathrm{I_2}}}\!\!\in\!\! S_2^{k_1},\mathrm{I_2}\in S_3^{k_2}, n_1\!\equiv\!k_2\!-\!m_2\!\}, \nonumber \\
&U_8\!=\!\{{\mathrm{I}}|{{\mathrm{I_2}}}\!\!\in\!\! S_3^{k_1},\mathrm{I_2}\in S_2^{k_2}, n_2\!\equiv\!k_1\!-\!m_1 \!\}, \nonumber \\
&U_9\!=\!\{{\mathrm{I}}|{{\mathrm{I_2}}}\!\!\in\!\! S_3^{k_1}\!,\mathrm{I_2}\in S_3^{k_2}\!, k_1\!-\!m_1\!\equiv\!k_2\!-\!m_2 \!\}, \nonumber \\
&U_{10}\!=\!\{{\mathrm{I}}|{{\mathrm{I_2}}}\!\!\in\!\!\{S_2^{k_1}\!\cup\! S_3^{k_1}\!\},\mathrm{I_2}\in\!\{S_2^{k_2}\!\cup\! S_3^{k_2}\}\!\} \nonumber \\& \ \ \ \ \ \ \ \ \ -\bigcup_{i\in\{6,7,8,9\}}\!\!\!\!\!\!U_i, \nonumber \\
&U_{11}\!=\!\{{\mathrm{I}}|{{\mathrm{I_2}}}\in\!\!S_5^{k_1},\mathrm{I_2}\in\!S_5^{k_2}\!, n_1\!-\!m_1\!\equiv\!k_2,\!n_2\!-\!m_2\!\equiv\!k_1\!\},
\end{align}
where
\begin{align}
&\mathbb{E}[P_{k_1}({{\mathrm{I_2}}})P^*_{k_2}\!(\mathrm{I_2})]\nonumber\\
&=\!\!\begin{cases} \frac{\psi_3^2 p_1^2}{K^2},   	\ \ \ \ \ {\mathrm{I}}\in \bigcup_{i\in\{1,4,5,10,11\}}U_i,\\
\frac{\psi_3^2 p_2}{K^2},  \ \ \ \ \ {\mathrm{I}} \in  \bigcup_{i\in\{2,3,6,7,8,9\}}U_i.
\end{cases}
\end{align}
We obtain $m_3$ as follows:
\begin{align}
\label{eq.subsetsum3}
&m_3=\mathbb{E}[{N_{\text{PA}}}[k_1]{N^*_{\text{PA}}}[k_2]]\nonumber\\
&=\sum_{i=1}^{10}|U_i|\mathbb{E}[P({{\mathrm{I_2}}})P(\mathrm{I_2})|{\mathrm{I}}\in U_i]\nonumber \\
&=\frac{\psi_{3}^2}{K^2}\left\{(4K^2-6K)p_1^2+4(K-1)p_2\right\}. 
\end{align}
\section{Proof of Theorem 2}
Let $\bold{G}$ be $\bold{N_{\text{PA}}}\circ \bold{H_{\text{SI}}}+\bold{N}$, where $\circ$ denotes the Hadamard product operation. Note that $P^{q,\text{PA}}_{H_{\text{eff}}}$ can be obtained by substituting $N$ in~\eqref{eq.residual_q2} with $G$. Hence, we express $P^{q,\text{PA}}_{H_{\text{eff}}}$ through modifying the last three terms in~\eqref{eq.residual_q2} (i.e., $c,d,e$) as follows:  
\begin{align}
\label{eq.49}
	P^{q,\text{PA}}_{H_{\text{eff}}} =& \frac{1}{K}{\mathrm{tr}}[\bold{E_{H_{{\mathrm{SI}}}H_{{\mathrm{SI}}}^*}}]\!-\!(a)\!+\!(b)\!+\!(c')\!+\!(d')\!-\!(e'), 
\end{align}
where
\begin{align}
\label{eq.cde_changed}
(c') &=\frac{1}{K}\text{tr}(\mathbb{E}[\bold{\Omega} \bold{R}^{-1}\bold{Q^*RQ}\bold{R}^{-1}\bold{\Omega}^*\bold{GG^*}]) \nonumber \\
(d') &=\frac{1}{K}\text{tr}(\mathbb{E}[\bold{\Omega} \bold{R}^{-1}\bold{Q^*RQ}\bold{R}^{-1}\bold{\Omega}^*(\bold{GH^*+HG^*})])\nonumber \\
(e') &=\frac{2}{K}\text{tr}(\mathbb{E}[\bold{\Omega} \bold{Q}\bold{R}^{-1}\bold{\Omega}^* \bold{G}\bold{H_{\text{SI}}}^*]).
\end{align} 
By using Lemma 1 and 2, we obtain
\begin{align}
\label{eq.Nsubst}
\mathbb{E}[\bold{GH^*_\text{SI}}]&=m_1\mathbb{E}[\bold{H_\text{SI}H^*_\text{SI}}], \nonumber \\
\mathbb{E}[\bold{GG^*}]=(m_2\!-\!m_3)&\bold{D(E_{H_{\mathrm{SI}}H_{\mathrm{SI}}^*})}+m_3\bold{E_{H_{\mathrm{SI}}H_{\mathrm{SI}}^*}}+\sigma^2 I,\nonumber\\
\end{align}
where $\bold{D(E_{H_{\mathrm{SI}}H_{\mathrm{SI}}^*})}$ denotes the diagonal matrix whose diagonal elements are equal to that of $\bold{E_{H_{\mathrm{SI}}H_{\mathrm{SI}}^*}}$.
Using~\eqref{eq.Nsubst}, we simplify~\eqref{eq.cde_changed} as follows:
\begin{align}
\label{eq.cde_changed2}
(c') &=(c)+m_3(b)\nonumber \\
&+\frac{m_2-m_3}{K}\left[P^2A_1^2K\text{tr}\left(\bold{\Omega} \bold{R}^{-1}\bold{\Omega}^*\bold{D(E_{H_{\mathrm{SI}}H_{\mathrm{SI}}^*})} \right)\right.\nonumber \\
& \ \left.+(A_2-P^2A_2^2)\text{tr}\left(\bold{\Omega} (\bold{R}^{-1})^2\bold{\Omega}^*\bold{D(E_{H_{\mathrm{SI}}H_{\mathrm{SI}}^*})}\right)\right] \nonumber \\
(d') &=2m_1(b)\nonumber\\
(e') &=m_1(a).
\end{align} 
Now we get Theorem 2 by substituting \eqref{eq.49} with \eqref{eq.cde_changed2}.

\section{Proof of Corollary 3}
By substituting $\bold{\Omega} \bold{R}^{-1} \bold{\Omega}^* \bold{E_{H_{{\mathrm{SI}}}H_{{\mathrm{SI}}}^*}}\!=\! \bold{E_{H_{{\mathrm{SI}}}H_{{\mathrm{SI}}}^*}}$, $\bold{R}=K\bold{I_M}$,  and $\sigma^2=0$, we can rewrite~\eqref{eq.cor3} as 
\begin{align}
\label{eq.cor3_proof}
&	\tilde{P}^{q,\text{PA}}_{H_{\text{eff}}} = \frac{1}{K}{\mathrm{tr}}[\bold{E_{H_{{\mathrm{SI}}}H_{{\mathrm{SI}}}^*}}]\nonumber \\
+&\frac{(1\!\!-\!\!2m_1\!+\!m_3)(PA_1)^2\!-\!2(m_1\!+\!1)\!PA_1}{K}{\mathrm{tr}}[\bold{E_{H_{\mathrm{SI}}H_{\mathrm{SI}}^*}}] \nonumber\\ 
+&\frac{\left(1\!\!-\!\!2m_1\!+\!m_3)(A_2-(PA_1)^2\right)}{K}{\mathrm{tr}}\left[\bold{E_{H_{\mathrm{SI}}H_{\mathrm{SI}}^*}}\right] \nonumber \\
+&\frac{A_2\left(m_2\!-\!m_3)\right)}{K^2}{\mathrm{tr}}\left[\bold{\Omega}\bold{\Omega}^* \bold{D(E_{H_{\mathrm{SI}}H_{\mathrm{SI}}^*})} \right].
\end{align}
We can easily get~\eqref{eq.cor3} from~\eqref{eq.cor3_proof} by using~\eqref{eq.cor2}.
\bibliographystyle{IEEEtran}
\bibliography{TWC}


\end{document}